\documentclass{article}

\newenvironment{proofof}[1]{\noindent \textbf{Proof of #1.}}{\qed}

\title{Learning Partitions with Optimal Query and Round Complexities}

\author{
Hadley Black \vspace{-.4em} \\
University of California, San Diego \vspace{-.4em} \\ 
{\tt hablack@ucsd.edu}
\and 
Arya Mazumdar \vspace{-.4em} \\ 
University of California, San Diego \vspace{-.4em} \\
{\tt arya@ucsd.edu}
\and 
Barna Saha\thanks{All authors supported by NSF TRIPODS Institute grant 2217058 (EnCORE) and NSF 2133484.} \vspace{-.4em} \\ 
University of California, San Diego \vspace{-.4em} \\
{\tt bsaha@ucsd.edu}
}

\usepackage{a4}
\usepackage{array}
\usepackage{tabularx}
\usepackage{graphicx}
\usepackage{amsmath,amssymb,amsthm,mathtools}
\usepackage{paralist}
\usepackage{bm}
\usepackage{bbm}
\usepackage{xspace}
\usepackage{url}
\usepackage{fullpage, prettyref}
\usepackage{boxedminipage}
\usepackage{wrapfig}
\usepackage{ifthen}
\usepackage{color}
\usepackage{xcolor}
\usepackage{framed}
\usepackage[pagebackref,letterpaper=true,colorlinks=true,pdfpagemode=none,urlcolor=blue,linkcolor=blue,citecolor=purple,pdfstartview=FitH,linktocpage = true]{hyperref}
\usepackage[nameinlink]{cleveref}
\usepackage{fullpage}
\usepackage{esvect}
\usepackage{mdframed}
\usepackage{thmtools}
\usepackage{thm-restate}
\newtheorem{theorem}{Theorem}[section]

\newtheorem{lemma}[theorem]{Lemma}
\newtheorem{claim}[theorem]{Claim}
\newtheorem{corollary}[theorem]{Corollary}

\newcommand{\ignore}[1]{}

\newcommand{\cA}{{\cal A}}

\newcommand{\cC}{{\cal C}}

\newcommand{\cE}{{\cal E}}

\newcommand{\cP}{\mathcal{P}}

\newcommand{\eps}{\varepsilon}

\newcommand{\poly}{\mathrm{poly}}

\newcommand{\ceil}[1]{\lceil#1\rceil}

\newcommand{\floor}[1]{\lfloor#1\rfloor}

\usepackage[english]{babel}
\usepackage[autostyle, english = american]{csquotes}
\MakeOuterQuote{"}
\newcommand{\Sec}[1]{\hyperref[sec:#1]{\Cref*{sec:#1}}} %section
\newcommand{\Eqn}[1]{\hyperref[eq:#1]{(\ref*{eq:#1})}} %equation
\newcommand{\Fig}[1]{\hyperref[fig:#1]{Fig.\,\ref*{fig:#1}}} %figure
\newcommand{\Tab}[1]{\hyperref[tab:#1]{Tab.\,\ref*{tab:#1}}} %table
\newcommand{\Thm}[1]{\hyperref[thm:#1]{Theorem\,\ref*{thm:#1}}} %theorem
\newcommand{\Fact}[1]{\hyperref[fact:#1]{Fact\,\ref*{fact:#1}}} %fact
\newcommand{\Lem}[1]{\hyperref[lem:#1]{Lemma\,\ref*{lem:#1}}} %lemma
\newcommand{\Prop}[1]{\hyperref[prop:#1]{Prop.~\ref*{prop:#1}}} %property
\newcommand{\Cor}[1]{\hyperref[cor:#1]{Corollary~\ref*{cor:#1}}} %corollary
\newcommand{\Conj}[1]{\hyperref[conj:#1]{Conjecture~\ref*{conj:#1}}} %conjecture
\newcommand{\Def}[1]{\hyperref[def:#1]{Definition~\ref*{def:#1}}} %definition
\newcommand{\Alg}[1]{\hyperref[alg:#1]{Alg.~\ref*{alg:#1}}} %algorithm
\newcommand{\Obs}[1]{\hyperref[obs:#1]{Obs.~\ref*{obs:#1}}} %observation
\newcommand{\Ex}[1]{\hyperref[ex:#1]{Ex.~\ref*{ex:#1}}} %example
\newcommand{\Clm}[1]{\hyperref[clm:#1]{Claim~\ref*{clm:#1}}} %example
\newcommand{\Step}[1]{\hyperref[step:#1]{Step~\ref*{step:#1}}} %example

\usepackage{hyperref}       % hyperlinks
\usepackage{url}            % simple URL typesetting
\usepackage{booktabs}       % professional-quality tables
\usepackage{amsfonts}       % blackboard math symbols
\usepackage{nicefrac}       % compact symbols for 1/2, etc.
\usepackage{xcolor}         % colors
\usepackage[ruled,linesnumbered]{algorithm2e}
\SetKwComment{Comment}{/* }{ */}
\usepackage{wrapfig}
\usepackage{algorithmic}
\usepackage{float}
\newcommand{\LS}{\mathsf{LearnSparse}}
\newcommand{\LRPQ}{\mathsf{LR}\text{-}\mathsf{SameSetQuery}}
\newcommand{\LRSQ}{\mathsf{LR}\text{-}\mathsf{WeakSubsetQuery}}
\newcommand{\LRSQS}{\mathsf{LR}\text{-}\mathsf{StrongSubsetQuery}}
\newcommand{\NASQ}{\mathsf{NA}\text{-}\mathsf{WeakSubsetQuery}}
\newcommand{\NASQS}{\mathsf{NA}\text{-}\mathsf{StrongSubsetQuery}}

\newcommand{\Ot}{\widetilde{O}}
\newcommand{\cK}{\mathcal{K}}
\newcommand{\tC}{\widetilde{\mathcal{C}}}

\newcommand{\same}{\mathsf{same}\text{-}\mathsf{set}}
\newcommand{\countq}{\mathsf{count}}
\newcommand{\partitionq}{\mathsf{partition}}

\begin{document}

\maketitle

\begin{abstract}
  We consider the basic problem of learning an unknown partition of $n$ elements into at most $k$ sets using simple queries that reveal information about a small subset of elements. Our starting point is the popular and well-studied pairwise \emph{same-set} queries which ask if a pair of elements belong to the same class. It is well-known that non-adaptive (fully parallel) algorithms require $\Theta(n^2)$ queries, while adaptive (fully sequential) algorithms require $\Theta(nk)$ queries, and the best known algorithm uses $k-1$ rounds of adaptivity. Many variations of this problem have been studied over the last two decades in multiple disciplines due to its fundamental nature and connections to clustering, active learning, and crowd-sourcing. In many of these applications, it is of paramount interest to reduce adaptivity, a.k.a the number of rounds, while minimizing the query complexity. In this paper, we give a complete characterization of the deterministic query complexity of this problem as a function of the number of rounds, $r$, which interpolates smoothly between the non-adaptive and adaptive settings: for any constant $r \geq 1$, the query complexity is $\smash{\Theta(n^{1+\frac{1}{2^r-1}}k^{1-\frac{1}{2^r-1}})}$. Additionally, our algorithm only needs $O(\log \log n)$ rounds to attain the optimal $O(nk)$ query complexity, which is a double-exponential improvement over prior works when $k$ is a polynomial in $n$.

Next, we consider two natural generalizations of pair-wise queries to general subsets $S$ of size at most $s$: (1) weak subset queries which return the number of classes intersected by $S$, and (2) strong subset queries which return the entire partition restricted on $S$. Once again in crowd sourcing applications, queries on large sets may be prohibitive. For non-adaptive algorithms, we show $\Omega(n^2/s^2)$ strong queries are needed. In contrast, perhaps surprisingly, we show that there is a non-adaptive randomized algorithm using weak queries that matches this bound up to log-factors for all $s \leq \sqrt{n}$. More generally, we obtain nearly matching upper and lower bounds for algorithms using weak and strong queries in terms of both the number of rounds, $r$, and the query size bound, $s$.
\end{abstract}

\newpage
\tableofcontents
\newpage

%\begin{keywords}%
%  Partition learning, clustering, query complexity, round complexity
%\end{keywords}

\section{Introduction}

Learning set partitions is fundamental to many applications including unsupervised learning, entity resolution, and network reconstruction. We consider the basic algorithmic problem of learning an unknown partition $\cP = (C_1,\ldots,C_k)$ of a universe $U$ of $n$ elements via access to an oracle that provides information about a queried subset $S \subseteq U$. Our starting point is the \emph{pairwise same-set} query oracle, which returns whether a queried pair of elements $\{x,y\} \subset U$ belong to the same class. As far as we know, partition learning under same-set queries goes back at least to  \cite{RS07} who considered the problem of learning the connected components of a graph. The problem was later introduced more broadly in the learning theory and clustering literature by independent works of 
\cite{ashtiani2016clustering},  \cite{MS17a,MS17b,mazumdar2017theoretical}, and  \cite{mitzenmacher2016predicting}. In parallel, the problem garnered interest in the database community as an important primitive to develop crowd-sourced databases  \cite{feng2011crowddb, mozafari2012active,DBLP:journals/tods/DavidsonKMR14}. Since then, it has been studied extensively over the last decade \cite{saha2019correlation, huleihel2019same,bressan2020exact,del2022clustering,DMT24} due to its fundamental nature and relevance to clustering, machine learning, and databases.
%see e.g. \cite{cohn2003semi,balcan2008clustering, bair2013semi}. 

In particular, the problem naturally models a situation where one would like to learn a hidden ground-truth clustering when obtaining explicit class labels is difficult or unnecessary, but deciding whether two elements have the same class label is easy. Indeed, learning partitions in this model can be viewed as a \emph{label-invariant} form of clustering. In many applications class labels are difficult to discern computationally due to noisy and incomplete data, but different classes are trivial to distinguish with the human eye. Thus, applications have been developed which implement clustering via same-set queries, where external crowd-workers play the role of the oracle, e.g. \cite{franklin2011crowddb, crowder12}. These queries are also straightforward to implement \cite{mazumdar2017semisupervised}, making them broadly applicable. 
%We provide a further review of the relevant literature in \Cref{sec:related-work}. 
Beyond its many motivating applications, we believe that partition learning with same-set queries is an extremely fundamental algorithmic problem akin to basic questions like comparison-based sorting, and as such deserves a thorough theoretical study.

\paragraph{Learning partitions with pairwise queries in few rounds.} In crowdsource clustering applications, parallelization of queries is paramount to minimizing execution time, since one may not have control over how long it takes for queries to be answered \cite{gu2012towards}. In the context of query-algorithms, parallelism is formalized in terms of \emph{round-complexity}: an algorithm has round-complexity $r$ if its queries can be batched into $r$ rounds $Q_1,\ldots,Q_r \subset U^2$ where the round $t$ queries, $Q_t$, are made all at once (in parallel) and only depend on the oracle's response to the queries in the previous rounds. We say that such an algorithm uses $r$ \emph{rounds of adaptivity}. An algorithm is called \emph{non-adaptive} if $r=1$ and \emph{fully adaptive} if no bound is given for $r$. It is well-known that the same-set query complexity of partition learning is $\Theta(nk)$ for fully adaptive algorithms \cite{RS07,DBLP:journals/tods/DavidsonKMR14, MS17a,LM22} and $\Theta(n^2)$ for non-adaptive algorithms \cite{MS17a,BLMS24}. However, despite numerous works studying same-set queries, and the clear motivation for simultaneously minimizing adaptivity and query complexity, there has not yet been a study of the round-complexity of this basic question. In particular, the best known algorithm achieving query complexity $O(nk)$ uses $k-1$ rounds \cite{RS07}. On the other hand, the $O(n^2)$ non-adaptive upper bound comes by trivially querying every pair of elements, and this is provably optimal, even when there are only $k=3$ sets in the partition. We ask, \emph{can this be significantly improved on using few rounds?}

Our work fills this gap: for every constant $r \geq 1$, we show matching upper and lower bounds of $\Theta(n^{1+\frac{1}{2^r-1}}k^{1-\frac{1}{2^r-1}})$, interpolating smoothly between the non-adaptive and fully adaptive settings. Moreover, our algorithm attains the optimal query complexity using only $O(\log \log n)$ rounds. As previous state of the art algorithms use $O(k)$ rounds, this is a double-exponential improvement over prior works when $k = \poly(n)$. Technically speaking, our algorithm uses a simple recursive framework which may be useful for improving round-complexity for other problems. Our lower bound uses a novel application of Tur\'{a}n's theorem to construct hard instances, and we believe our techniques may be useful for proving lower bounds for other query-based algorithmic problems for graphs and other combinatorial objects. See \Cref{sec:results-pairwise} for explicit statements and further discussion on our results for same-set queries. 

%Do you see any place where the technical components in this work can be useful?

%\comment{adaptivity and parallelism}

\paragraph{Generalizing to subsets while minimizing query size and rounds.} Next, we consider generalizations of the popular pairwise same-set query model to general subsets $S \subseteq U$. Beyond our specific problem, there has been recent interest in the learning theory community in learning concepts using subset, or group queries, e.g. \cite{KMT24}. Recent works \cite{CL24,BLMS24} introduced the partition learning problem with access to an oracle returning the number of sets in the partition intersecting the query, $S \subseteq U$. An information-theoretic lower bound is $\Omega(n)$ and \cite{CL24} obtained a matching $O(n)$ query adaptive algorithm, while for non-adaptive algorithms \cite{BLMS24} showed that $O(n \log k \cdot (\log k + \log\log n)^2)$ queries is possible. From a practical perspective, an obvious downside to these algorithms is that queried subsets can be large: \cite{CL24} uses $O(k)$ sized queries and \cite{BLMS24} uses $O(n)$ sized queries. To address this, we investigate the query complexity of learning partitions when a bound of $s$ is placed on the allowed query size. In both the adaptive and non-adaptive cases, we obtain algorithms with the same query complexity up to logarithmic factors, while shrinking the query size \emph{quadratically}, which is also optimal. Moreover, we obtain nearly matching upper and lower bounds for query size $s$ in terms of \emph{round-complexity}, similar to our results for pairwise same-set queries. 

To motivate subset queries that count the number of intersected classes (weak subset queries), we also study the strongest possible type of subset query, which returns a full representation of the partition on the queried subset (strong subset queries). Surprisingly, we show that the number of weak vs. strong subset queries that are required is the same, up to logarithmic factors, for all $s$ up to a reasonably large threshold ($s \leq O(\sqrt{n})$ for non-adaptive and $s \leq O(\sqrt{k})$ for adaptive), while weak queries require \emph{significantly} less communication. A detailed discussion of our results for subset queries is given in \Cref{sec:results-subsets}.

\paragraph{Organization.} All of our results are summarized in \Cref{sec:results}. We prove our lower and upper bounds for pairwise queries in \Cref{sec:LB} and \Cref{sec:LR-pair}, respectively. Our results for weak and strong subset queries are proven in \Cref{sec:count} and \Cref{sec:LR-partition}, respectively. Our most technical results are the lower bound for pairwise queries and the non-adaptive algorithm using weak subset queries. Informal overviews for these proofs are given in \Cref{sec:LB-overview} and \Cref{sec:NA-subset-overview}, respectively.

%Our most interesting and technical results are the lower bound of \Cref{sec:LB} and the non-adaptive algorithm for weak subset queries of \Cref{sec:weak}.

\subsection{Results} \label{sec:results}

\subsubsection{Pairwise Queries} \label{sec:results-pairwise}

We first consider the basic problem of learning an arbitrary unknown $k$-partition\footnote{We use the term $k$-partition as shorthand for a partition into \emph{at most} $k$ sets.} of $n$ elements using pairwise same-set queries: given $x,y \in U$ and a hidden partition $\cP$, $\same(x,y,\cP) = \mathsf{yes}$ if $x,y$ belong to the same set in $\cP$ and $\same(x,y,\cP) = \mathsf{no}$ otherwise.

Our focus is on \emph{round-complexity}: we first design a simple $r$-round deterministic algorithm which attains the optimal $O(nk)$ query complexity using only $O(\log \log n)$ rounds, which is a double-exponential improvement over the previous best $k-1$ round algorithm when $k = \poly(n)$. In general, with $O(\log 1/\eps)$ rounds, our algorithm has query complexity $O(n^{1+\eps}k^{1-\eps})$, which interpolates smoothly between the non-adaptive and fully adaptive settings.

\begin{theorem} [Pair query upper bound] \label{thm:LR-pair-UB} For any $r,k \geq 1$, there exists a deterministic $r$-round algorithm for $k$-partition learning using at most $8n^{(1 + \frac{1}{2^{r}-1})} k^{(1 - \frac{1}{2^{r}-1})}$ pairwise same-set queries.
\end{theorem}

In fact, we show that our algorithm attains the optimal interpolation between the non-adaptive and fully adaptive, $O(\log \log n)$ round setting, up to a factor of $r$. In particular, our upper and lower bounds are exactly matching for every constant $r$, and are always tight up to a $O(\log \log n)$ factor.

\begin{theorem} [Pair query lower bound] \label{thm:LR-pair-LB} For all $r \geq 1$, any $r$-round deterministic algorithm for $k$-partition learning must use at least $\Omega\big(\frac{1}{r} \cdot n^{(1+\frac{1}{2^r-1})}k^{(1-\frac{1}{2^r-1})}\big)$ pairwise same-set queries.
\end{theorem}

We remark that it is still open to establish such a lower bound for arbitrary \emph{randomized} algorithms, and we believe that additional technical ideas are needed to achieve such an extension.

%In particular, our results show that for every constant $r \geq 1$, the pair-wise query-complexity of $k$-partition learning is $\Theta\big(n^{(1+\frac{1}{2^r-1})}k^{(1-\frac{1}{2^r-1})}\big)$

%\[
%\Theta\left(n^{\left(1 + \frac{1}{2^{r}-1}\right)} k^{\left(1 - \frac{1}{2^{r}-1}\right)}\right) \text{.}
%\]

\subsubsection{Subset Queries} \label{sec:results-subsets}

Next, we consider the two following generalizations of pairwise same-set queries to subsets. Given hidden partition $\cP$ and a queried subset $S \subseteq U$, each oracle returns the following information.
\begin{itemize}%[leftmargin=*]\setlength\itemsep{0em}
    \item \textit{Weak subset query oracle:} Given $S \subseteq U$, the oracle returns $\countq(S,\cP) := \sum_{X \in \cP} \mathbf{1}(S \cap X \neq \emptyset),$
    %|\{S \cap X \neq \emptyset \colon X \in \cP\}|$
     i.e. the number of parts which $S$ intersects.\footnote{Weak subset queries are also sometimes referred to as "rank queries", e.g. \cite{CL24}, or simply as "subset queries", e.g. \cite{BLMS24}.}
    \item \textit{Strong subset query oracle:} Given $S \subseteq U$, the oracle returns $\partitionq(S,\cP) := \{S \cap X \colon X \in \cP\}$, i.e. a full description of $\cP$ restricted on $S$.
\end{itemize}

We are interested in the query complexity of learning partitions when an upper bound of $s \in [2,n]$ is placed on the allowed size of a queried subset. Strong queries are the \emph{most informative} type of subset query that one can define for the partition learning problem and thus provide a meaningful benchmark against which to measure the effectiveness of other query types. When $s=2$, both queries are equivalent (in fact they are the same as pairwise queries). At the other extreme, when $s = n$ a single strong query recovers the entire partition, while a weak query only returns the number of parts. Intuitively, as $s$ increases, the gap between weak and strong queries widens. Given that weak queries require significantly less communication from the oracle ($O(\log k)$ as opposed to $O(s \log k)$ bits), as well as less computation to answer, a motivating question for this line of inquiry is: \emph{is there a regime of} $s \gg 2$ \emph{where weak and strong queries are similarly powerful?}

%On the other extreme, when $s = n$, a single strong subset query on $U$ returns the entire partition, whereas a single weak subset query on $U$ will simply return the number of sets in the partition. 

%The respective information content of strong and weak queries combined with the fact that there are $k^{\Omega(n)}$ partitions possible, immediately implies an $\Omega(n)$ lower bound for weak queries (regardless of $s$) and (b) an $\Omega(n/s)$ lower bound for strong queries. 

%\begin{observation} [Information-theoretic lower bounds] \label{obs:info-LB} The partition learning problem using $s$-bounded subset queries requires (a) $\Omega(n)$ weak queries and (b) $\Omega(n/s)$ strong queries. \end{observation}

%However, when $s$ is sufficiently small, stronger bounds are possible using the

\paragraph{The non-adaptive case.} 
Obtaining lower bounds, even for strong subset queries, is straightforward using known lower bounds for pairwise queries: observe that one $s$-bounded strong subset query can be simulated (non-adaptively) by ${s \choose 2}$ pair-wise same-set queries\footnote{Given a set $S$, one can query all pairs in $S$ and compute the entire partition restricted on $S$.}. %, since given a set $S$, one can trivially query all pairs in $S$, compute the partition restricted on $S$, and then return the number of sets. 
%Thus, any lower bound for pairwise queries extends to both types of subset queries of size at most $s$, with a multiplicative factor of $1/s^2$. 
Thus, for non-adaptive algorithms, the $\Omega(n^2)$ lower bound for pairwise queries \cite{MS17a,BLMS24} implies that $\Omega(n^2/s^2)$ strong queries are necessary. In fact, we prove (in \Cref{sec:LR-partition}) that there is also a simple deterministic non-adaptive algorithm matching this bound.

\begin{theorem} [Non-adaptive strong queries]\label{thm:NA-strong} For $s \in [2,n]$, the non-adaptive strong query complexity of partition learning is $\Theta(n^2/s^2)$. The algorithm is deterministic and the lower bound holds even for randomized algorithms. \end{theorem}

On the other hand, a weak subset query contains at most $O(\log k)$ bits of information and there are $k^{\Omega(n)}$ partitions possible, implying an information-theoretic lower bound of $\Omega(n)$, even when $s = n$. Therefore, there is a separation between weak and strong queries when $s \gg \sqrt{n}$, but this leaves as a possibility that weak and strong queries could have similar power when $s = O(\sqrt{n})$.  
%A natural question is then whether one can also match this bound (up to log-factors) using weak queries when $s \leq O(\sqrt{n})$. 
Previous work of \cite{BLMS24} provided two algorithms making $\Ot(n^2k/s^2)$ and $\Ot(n^2/s)$ queries, respectively, but it remained open whether $\Ot(n^2/s^2)$ is possible. Our main result for non-adaptive subset queries provides an affirmative answer to this open question. (Question 1.14 of \cite{BLMS24}.) %, showing that there is no separation (up to log-factors) between weak and strong queries when $s \in [2,\sqrt{n}]$. 

%\begin{itemize}
%    \item For weak, $s$-bounded subset queries, non-adaptive algorithms require $\Omega(\max(\frac{n^2}{s^2}, n))$ queries and fully adaptive algorithms require $\Omega(\max(\frac{nk}{s^2},n))$ queries. 
%    \item For strong, $s$-bounded subset queries, non-adaptive algorithms require $\Omega(\frac{n^2}{s^2})$ queries and fully adaptive algorithms require $\Omega(\max(\frac{nk}{s^2},\frac{n}{s}))$ queries. 
%\end{itemize}

%This begs the question of whether there is in fact a near-linear query algorithm using only $O(\sqrt{n})$ sized queries. In fact, \cite{BLMS24} obtained a non-adaptive algorithm making $\widetilde{O}(\frac{n^2k}{s^2})$ queries for all $s \leq \sqrt{n}$, but it remained open whether the lower bound could be matched for $k \gg \poly\log(n)$. We resolve this open question by providing a nearly-optimal non-adaptive algorithm for all $s$, regardless of $k$.

\begin{theorem} [Non-adaptive weak queries, \Cref{thm:NA-bounded} informal] For all $s \in [2,\sqrt{n}]$, the weak subset query complexity of partition learning is $\widetilde{\Theta}(n^2/s^2)$. Our algorithm is randomized and succeeds with probability $1-1/\poly(n)$. \end{theorem}

We find this result to be surprising since intuitively strong queries seem to be \emph{significantly} more informative, yet up to logarithmic factors they provide no advantage for $s \leq \sqrt{n}$. From an applications perspective, this provides a compelling case for weak subset queries as they are nearly as informative as the strongest possible type of subset query, while being both (a) simpler to answer and (b) requiring significantly less communication.

\paragraph{The adaptive case.} Again, since one strong subset query of size $s$ can be simulated using $O(s^2)$ pairwise queries, \emph{any} lower bound for pairwise query algorithms extends to subset queries with an additional $1/s^2$ factor. Thus the $\Omega(nk)$ lower bound for fully adaptive pair query algorithms \cite{MS17a,LM22} extends to an $\Omega(nk/s^2)$ lower bound for strong subset queries. More generally, our lower bound \Cref{thm:LR-pair-LB} extends in the same fashion. Moreover, we use our non-adaptive subset query algorithms above combined with the algorithmic strategy used to obtain our $r$-round pair query algorithm of \Cref{thm:LR-pair-UB} to prove the following bounds on the query complexity of $r$-round subset query algorithms for partition learning.

\begin{theorem} [$r$-round weak subset queries] For every $s \geq 2$, $k \geq 1$, and $r \geq 1$, there is a randomized $r$-round algorithm for $k$-partition learning that succeeds with probability $1-1/\poly(n)$ using 
\[
\Ot\left(\max\Big(\frac{1}{s^2} \cdot n^{(1+\frac{1}{2^r-1})}k^{(1-\frac{1}{2^r-1})}, n \Big)\right)
\]
$s$-bounded weak subset queries and any $r$-round deterministic algorithm for this task must use 
\[
\Omega\left(\max\Big(\frac{1}{r} \cdot\frac{1}{s^2} \cdot n^{(1+\frac{1}{2^r-1})}k^{(1-\frac{1}{2^r-1})}, n \Big)\right) 
\]
$s$-bounded weak subset queries. \end{theorem}

Our upper and lower bounds are tight for constant $r$. Using $r = O(\log \log n)$ rounds and query size bound $s = O(\sqrt{k})$ our algorithm has nearly-linear query complexity, $\widetilde{O}(n)$. This is in contrast to \cite{CL24} who obtained an $O(n)$ query algorithm using $r = O(\log k)$ and $s = O(k)$. More generally, with $s = O(\sqrt{k^{1+\eps}})$ and $r = O(\log 1/\eps)$, our algorithm has query complexity $\Ot(n^{1+\eps})$.

We obtain similar results for strong subset queries. Note that an $s$-bounded strong query contains $O(s \log k)$ bits of information, implying an $\Omega(n/s)$ information-theoretic lower bound. 

\begin{theorem} [$r$-round strong subset queries] For every $s \geq 2$, $k \geq 1$, and $r \geq 1$, there is a deterministic $r$-round algorithm for $k$-partition learning using 
\[
O\left(\max\Big(\frac{1}{s^2} \cdot n^{(1+\frac{1}{2^r-1})}k^{(1-\frac{1}{2^r-1})}, \frac{n}{s} \Big)\right)
\]
$s$-bounded strong subset queries and any $r$-round deterministic algorithm for this task must use 
\[
\Omega\left(\max\Big(\frac{1}{r} \cdot\frac{1}{s^2} \cdot n^{(1+\frac{1}{2^r-1})}k^{(1-\frac{1}{2^r-1})}, \frac{n}{s} \Big)\right) 
\]
$s$-bounded strong subset queries.
\end{theorem}

Again, our bounds are tight for every constant $r$. Note that the query complexity of strong and weak subset queries is the same up to logarithmic factors when $s \leq \sqrt{n^{\frac{1}{2^r-1}}k^{1-\frac{1}{2^r-1}}}$, i.e. until the information-theoretic lower bound is reached for weak subset queries. However, strong subset queries require $\widetilde{\Omega}(s)$ bits of communication by the oracle and so the total communication is never less for strong query algorithms. Refer to \Fig{strongvsweak} for a visual comparison of strong vs. weak subset queries in terms of $s$ for non-adaptive and fully adaptive algorithms.

\begin{figure}
    \hspace*{0cm}
    \includegraphics[scale=0.50]{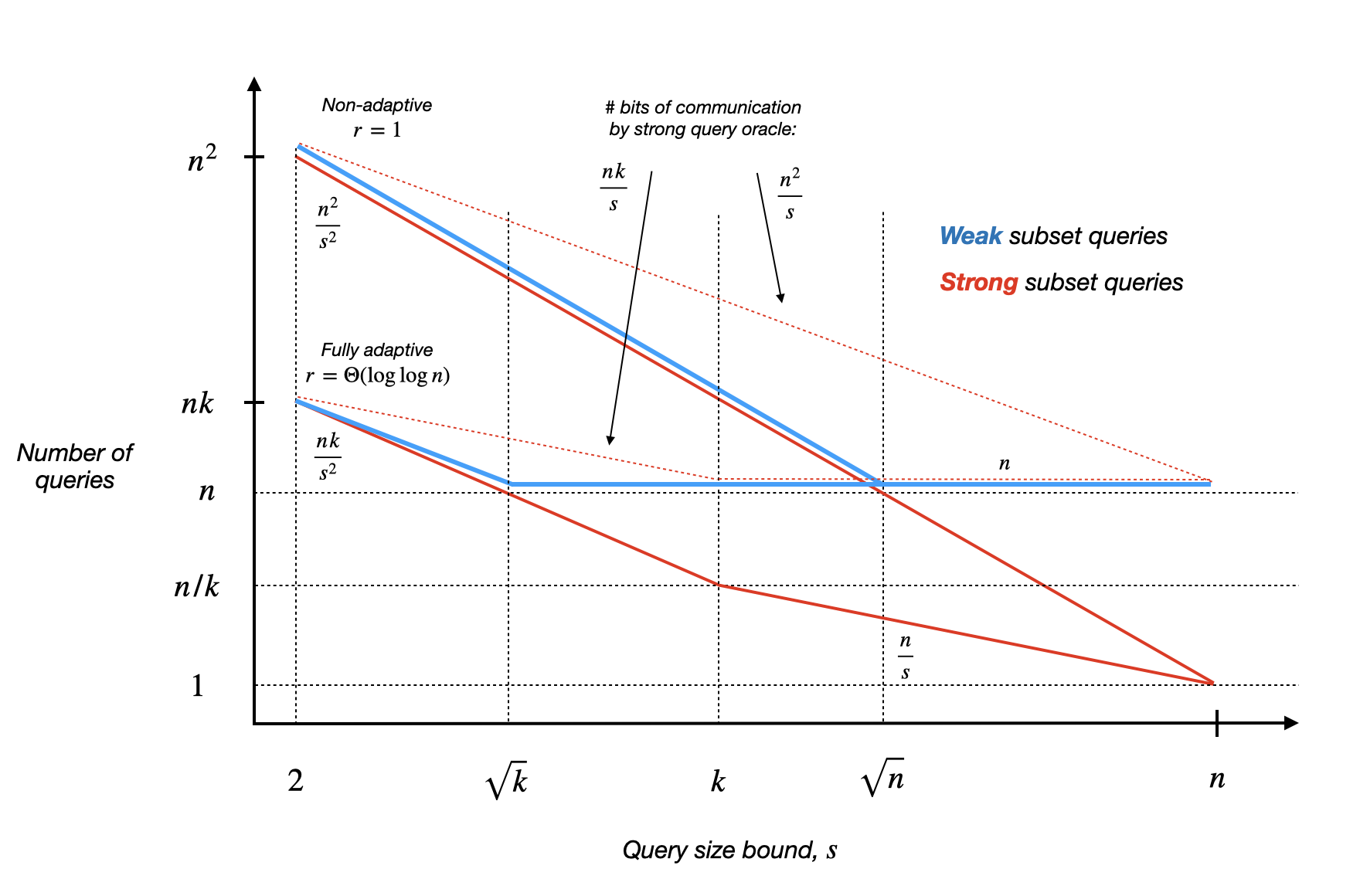}
    \caption{\small{A comparison of weak vs. strong subset queries for non-adaptive and fully adaptive algorithms as a function of the allowed query size, $s$, ignoring $\poly \log(n)$ factors. For the purposes of the diagram we have used $k \leq \sqrt{n}$, but note in general there is no such restriction on $k$. Our results reveal that in the relatively small $s$ regime ($s \leq \sqrt{n}$ for non-adaptive and $s \leq \sqrt{k}$ for adaptive), strong subset queries give no advantage over weak subset queries, which we find to be very surprising. Moreover, answering a weak query requires only $O(\log k)$ bits of communication by the oracle whereas answering a strong query requires $O(s \log k)$ bits of communication. Thus, in the weak query model the total communication follows the query complexity (the blue lines in the diagram), while in the strong query model the total communication is larger than the query complexity by a factor $O(s \log k)$ (pictured by the dotted red lines in the diagram). Thus, in terms of total communication, the weak query model is superior to the strong query model for all values of $s$. For clarity, we have chosen to picture the case of $r=1$ and $r=\Theta(\log \log n)$, but for any intermediate value of $r$, the query complexities follow a similar shape, lying in between these two cases. In general, the query complexities in the strong and weak models are the same for $s$ up to $\sqrt{n^{\eps(r)}k^{1-\eps(r)}}$, at which point the information-theoretic lower bound kicks in for weak queries. The strong query complexity continues to improve at the same rate up to $s = n^{\eps(r)}k^{1-\eps(r)}$, where the query complexity becomes $(n/k)^{1-\eps(r)}$, after which the information theoretic lower bound of $n/s$ kicks in. }} 
    \label{fig:strongvsweak}
\end{figure}

\subsection{Additional Related Work} \label{sec:related-work}

%\paragraph{Combinatorial search and graph query problems.}

Learning partitions with pairwise same-set queries is also sometimes cast in terms of learning the connected components of a graph where a query detects whether there exists a path between two vertices \cite{RS07,LM22}. Reconstruction of graphs from \emph{edge-detection queries} has also been studied, e.g. \cite{AB19}. The problem is also closely related to the well-studied problems of edge-sign prediction \cite{leskovec2010predicting, mitzenmacher2016predicting}, correlation clustering \cite{BBC04, ACN08, CMSY15, saha2019correlation, CCLLNV24}, and the stochastic block model \cite{abbe2018community, mukherjee2024recovering}. 

Our results for learning partitions with subset queries also fall into the area of \emph{combinatorial search} \cite{A88,DH00}, which studies query-based algorithms for reconstructing combinatorial objects. This field has been around since the 1960's, dating back to early works studying the \emph{coin-weighing problem} \cite{cantor1966determination,Lindstrom65}, i.e. reconstruction of Boolean vectors from additive queries, and \emph{group testing} \cite{HS87,du2000combinatorial, DH00,PR08,mazumdar2016nonadaptive,flodin2021probabilistic}). These results have been used as basic primitives to design optimal algorithms for graph reconstruction using additive and cross-additive queries (also commonly referred to as CUT queries) \cite{GK00,RS07,choi2008optimal,BM11b,BM11a,choi2013polynomial}. A special case of graph reconstruction that is closely related to our work is that of reconstructing matchings \cite{GK00,ABKRS04,AA05}, and in fact this connection was exploited by \cite{CL24} to obtain their optimal adaptive algorithm for partition learning with subset queries. The partition learning problem with weak subset queries can be viewed as learning a \emph{hyper-matching} using CUT queries \cite{CL24}. Beyond reconstruction problems, there is a growing body of work that studies query-based algorithms for various graph problems, such as computing minimum cuts \cite{rubinstein2017computing,ApersEGLMN22} and computing connectivity and spanning forests \cite{AL21,ApersEGLMN22,ChakrabartyL23,C24}.

\section{Lower Bound for Pair Queries} \label{sec:LB}

 %In particular, we will prove the following theorem.

%\begin{theorem} \label{thm:LB-rk} If $2 \leq r \leq k-2$, then any $r$-round deterministic pair-query algorithm for $k$-partition learning must have query complexity 
%\[
%    \Omega\left(n^{1+\frac{1}{2^r-1}}\left(k/r\right)^{1-\frac{1}{2^r-1}}\right)\text{.}
%\]
%\end{theorem}

%Before proving the theorem, we first provide the preliminary ideas and tools. At a high level, our approach is to view the queries as the edges of a graph, and exploit this graph's structure to construct hard instances, i.e. a pair of partitions that the algorithm fails to distinguish. The main structure that allows us to do this an independent set. Let $Q_1$ denote the set of queries made by the algorithm in the first round and suppose that there is a large independent set $Z$ in the corresponding graph $G(U,Q_1)$. Then, consider a partition $\cP = \{U \setminus Z\} \cup \cP_Z$ where $\cP_Z$ is some partition of $Z$. Then, the queries $Q_1$ give no information about $\cP_Z$, and so this is left up to the remaining rounds. This suggests a lower bound argument which performs an induction over the $r$ rounds of the algorithm to construct a hard-to-distinguish pair of partitions. 

In this section we prove our lower bound on $r$-round deterministic algorithms for learning $k$-partitions, \Cref{thm:LR-pair-LB}. As a warm-up, we begin by sketching an $\Omega(n^{1+\frac{1}{2^r-1}})$ lower bound, and then proceed to the main proof. At a high level, our approach is to view queries as edges of a graph, and exploit independent sets in this graph to construct hard instances, i.e. a pair of partitions that the algorithm fails to distinguish. To find large independent sets we use Tur\'{a}n's theorem, which we prove in \Cref{sec:deferred-proofs} for completeness.

\begin{theorem} [Tur\'{a}n's Theorem] \label{thm:turan} Let $G = (V,E)$ be an undirected graph with $n$ vertices and average degree $d_G$. Then $G$ contains an independent set of size at least $\frac{n}{1+d_{G}}$. \end{theorem}

\paragraph{Warm-up: an $\Omega(n^{1+\frac{1}{2^r-1}})$ using Tur\'{a}n's Theorem.} Let $r \leq \frac{1}{10} \log \log n$ and consider any $r$-round deterministic algorithm $\cA$. For brevity, we use $\eps(r) := \frac{1}{2^r-1}$ for all $r \geq 1$.

We will inductively construct a pair of $k$-partitions that $\cA$ fails to distinguish. The base case is when $r=1$. For this, suppose $\cA$ makes strictly fewer than ${n \choose 2}$ queries. Then there is a pair of points $x,y \in U$ for which $(x,y)$ is not queried. Then the $3$-partitions $(U \setminus \{x,y\},\{x,y\})$ and $(U \setminus \{x,y\},\{x\},\{y\})$ are not distinguished and this completes the base case. (I.e. the oracle returns the same answer for the two partitions on every query made by $\cA$.) 

Now, let $r \geq 2$ and $k \geq r+2$ and suppose for the sake of contradiction that $\cA$ makes fewer than $\frac{1}{3} \cdot n^{1+\eps(r)}$ queries in the first round. Let $G = (U,E)$ be the graph on $U$ whose edge-set is given by this set of queries. The average degree in this graph is at most $\frac{1}{3} \cdot n^{\eps(r)}$ and so by Tur\'{a}n's theorem, $G$ has an independent set $Z$ of size $|Z| \geq \frac{n}{1 + \frac{1}{3} \cdot n^{\eps(r)}} \geq n^{1-\eps(r)}$ by our upper bound assumption on $r$. In particular, in the first round, $\cA$ makes \emph{no queries whatsoever} in $Z$. Let $\cA_{r-1,Z}$ denote the algorithm using the remaining $r-1$ rounds of $\cA$ restricted on $Z$. Construct a pair of partitions by letting $U \setminus Z$ be a set in each, and within $Z$ inductively define the hard pair of partitions for $(r-1)$-round algorithms for learning a $(k-1)$-partition. By induction $\cA_{r-1,Z}$ must make at least
\[
\frac{1}{3} \cdot |Z|^{1+\eps(r-1)} \geq \frac{1}{3} \cdot \left(n^{1-\eps(r)}\right)^{1+\eps(r-1)} = \frac{1}{3} \cdot n^{\left(1-\eps(r)\right)\left(1+\eps(r-1)\right)} = \frac{1}{3} \cdot n^{1 + \eps(r)}
\]
queries, where in the last step we used item (1) of \Cref{clm:eps(r)}. Thus, we have a contradiction and this completes the proof. \\%\end{proof}

%\subsubsection{Strengthening the Lower Bound}

%\subsection{Strengthening the Lower Bound for $r$-Rounds: Proof of Theorem \ref{thm:LR-pair-LB}}
%\label{sec:LB-proof}

We now show how to strengthen the above argument to obtain an additional dependence on $k$, and ultimately prove \Cref{thm:LR-pair-LB}. First, it is not too hard to see that one can repeatedly apply Tur\'{a}n's theorem to obtain a collection of disjoint independent sets, giving the following Corollary, whose proof we defer to \Cref{sec:repeated-Turan-proof}. %This is a key technical ingredient for obtaining our improved lower bound. 

\begin{corollary} [Repeated Tur\'{a}n's Theorem] \label{cor:repeated-turan} Let $G = (V,E)$ be an undirected graph with $n$ vertices and $m \geq n$ edges. Let $N \leq n^2/8 m$ and $\ell \leq n/2N$. Then, $G$ contains $\ell$ disjoint independent sets, each of size $N$. \end{corollary}

\subsection{Proof Overview} \label{sec:LB-overview}

We now informally describe the proof of \Cref{thm:LR-pair-LB}. Let $\cA$ be an arbitrary deterministic $r$-round algorithm making $\ll n^{1+\eps(r)}(k/r)^{1-\eps(r)}$ queries, where $Q_t$ denotes the set of queries made in the $t$-th round. Note that $Q_1$ is a fixed, pre-determined set, but for $t \geq 2$, $Q_t$ depends on the oracle's responses to the queries in $Q_1 \cup \cdots \cup Q_{t-1}$. We need to show that there exists a pair of $k$-partitions that $\cA$ does not distinguish. To accomplish this we first show that there is a single partition $\cP$ such that after running the first $r-1$ rounds of $\cA$ on $\cP$, there still exists a "large" set $S$ (the size of $S$ will be $\approx \sqrt{n^{1+\eps(r)}(k/r)^{1-\eps(r)}}$) such that for \emph{every query $(u,v) \in Q_1 \cup \cdots \cup Q_{r-1}$ that touches $S$} (at least one of $u,v$ belong to $S$), $u$ and $v$ belong to different sets in $\cP$. Now, we can take any un-queried pair $(x,y) \in {S \choose 2}$ and define $\cP_{x,y}^{(1)}$, $\cP_{x,y}^{(2)}$ which modify $\cP$ by either making $\{x\}, \{y\}$ two separate sets or making $\{x,y\}$ a single set (see \Fig{LB2}). Crucially, this is \emph{well-defined} because all queried pairs involving $x$ or $y$ span different sets in $\cP$. I.e. the oracle's responses on all queries in the first $r-1$ rounds are consistent between the partitions $\cP$ and $\cP_{x,y}^{(b)}$. Finally, to distinguish $\cP_{x,y}^{(1)}$, $\cP_{x,y}^{(2)}$ the final round of queries $Q_r$ must contain the pair $(x,y)$. Thus, we can conclude $Q_r$ must contain every pair in ${S \choose 2} \setminus (Q_1 \cup \cdots \cup Q_{r-1})$, for otherwise there is some pair $\cP_{x,y}^{(1)}$, $\cP_{x,y}^{(2)}$ that is not distinguished by the algorithm. This shows that we must have $|Q_r| \approx n^{1+\eps(r)}(k/r)^{1-\eps(r)}$, contradicting the assumed upper bound on the number of queries.

Now, the main effort in the proof is in constructing the partition $\cP$ and the set $S$ described above. This is accomplished by inductively constructing a sequence of partitions $\cP_1,\cP_2,\ldots,\cP_{r-1}$ and sets $S^{(1)},S^{(2)},\ldots,S^{(r-1)}$ and taking $\cP := \cP_{r-1}$ and $S := S^{(r-1)}$. Essentially, the property we want each pair $(\cP_t, S^{(t)})$ in the sequence to have is what is described in the previous paragraph: every query made in the first $t$ rounds which touches $S^{(t)}$ is spanning different sets in $\cP_t$. This is done in the following manner. Invoke the repeated Tur\'{a}n theorem on $G(U,Q_1)$ to obtain $\ell \approx k/r$ disjoint independent sets $S^{(1)}_1,\ldots,S^{(1)}_{\ell}$ and let $S^{(1)}$ denote their union. Define $\cP_1 = \{U \setminus S^{(1)}, S^{(1)}_1,\ldots,S^{(1)}_{\ell}\}$ and observe that $(\cP_1,S^{(1)})$ has the desired property for the first round. Now, given $(\cP_{t-1},S^{(t-1)})$ with the desired property for the first $t-1$ rounds, we again invoke the repeated Tur\'{a}n theorem, but this time on $G(S^{(t-1)},Q_1 \cup \cdots \cup Q_{t-1})$ to obtain disjoint independent sets $S^{(t)}_1,\ldots,S^{(t)}_{\ell}$ and let $S^{(t)} \subseteq S^{(t-1)}$ denote their union. Now, starting from $\cP_{t-1}$, we construct $\cP_t$ which will define the oracle's response to the $t$-th round queries and which has the desired property for the first $t$ rounds. Each round introduces $\ell \approx k/r$ new sets into the partition and so at the end we have $\approx k$ sets. A subtle aspect of the argument is that one must be very careful to ensure that the oracle's responses on the first $t-1$ rounds of queries are consistent between $\cP_t$ and $\cP_{t-1}$. This is because the set $Q_t$ is determined by the oracle's responses to $Q_1 \cup \cdots \cup Q_{t-1}$ on $\cP_{t-1}$ which then define $\cP_t$. Thus, without this consistency property, this process would be ill-defined. (See \Fig{LB1} for an accompanying illustration of this argument.)

\begin{figure}
    \hspace*{2cm}
    \includegraphics[scale=0.45]{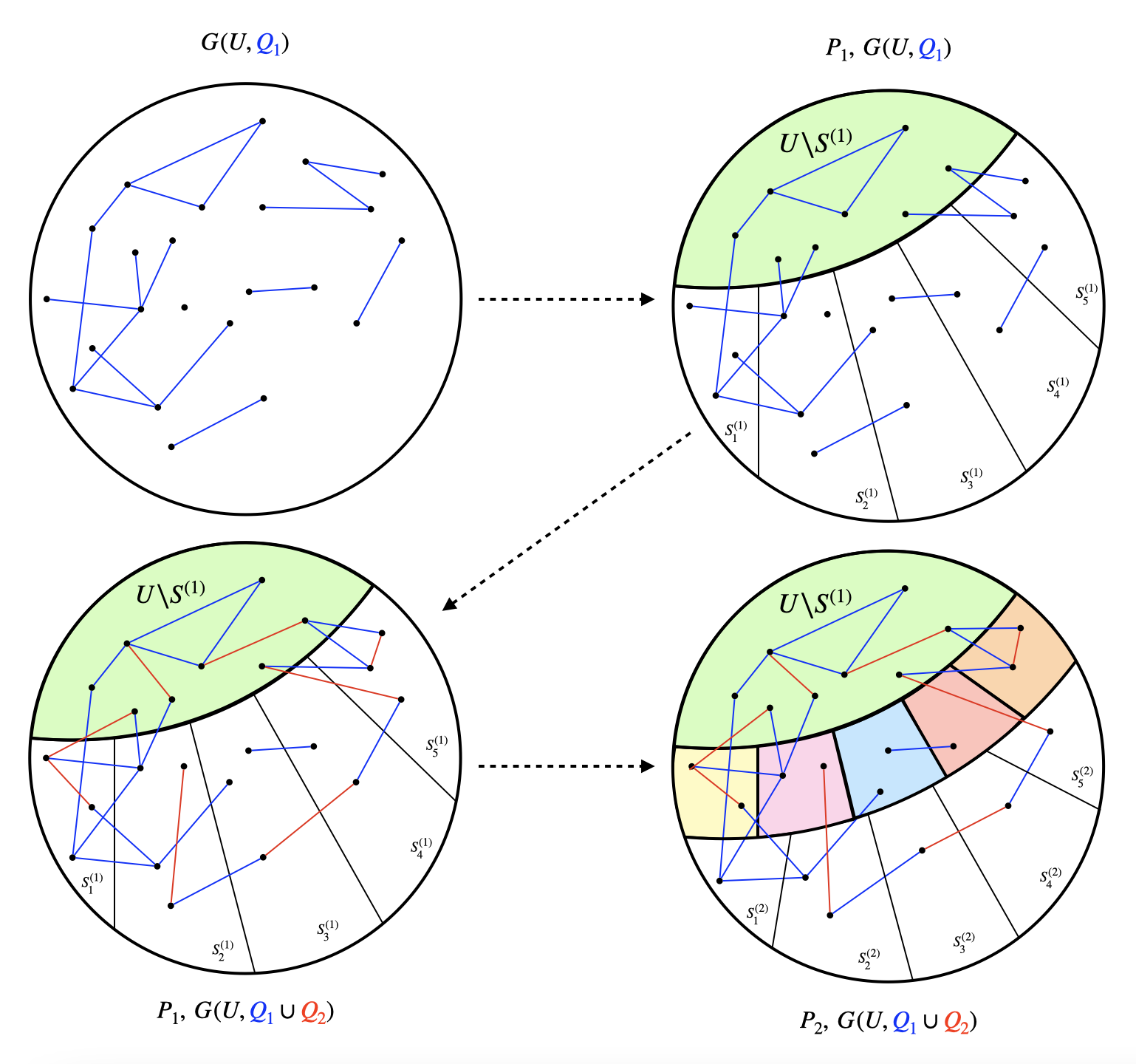}
    \caption{\small{An illustration depicting the construction of $\cP_1,S^{(1)}$ and $\cP_2,S^{(2)}$ in the proof of \Cref{thm:LR-pair-LB}.} The top-left shows the graph whose edges are the first round queries, $Q_1$ (blue edges). Then, \Cref{lem:carving} is applied which uses Tur\'{a}n's theorem to find $\ell$ (in the picture $\ell = 5$) independent sets $S^{(1)} = S^{(1)}_1 \sqcup \cdots \sqcup S^{(1)}_5 \subset U$ with respect to $Q_1$. The partition $\cP_1$ is defined based on these independent sets (top-right). Then, based on the oracle's responses to $Q_1$, the second round queries, $Q_2$, arrive (red edges, bottom-left). Again, \Cref{lem:carving} is applied to find $\ell$ independent sets $S^{(2)} = S^{(2)}_1 \sqcup \cdots \sqcup S^{(2)}_5 \subset S^{(1)}$ with respect to $Q_1 \cup Q_2$ and the partition $\cP_2$ is defined (bottom-right). The (non-green) shaded regions in the bottom-right represent the sets $S^{(1)}_j \setminus S^{(2)}$ for each $j \in [\ell]$. The construction repeats in this way for $r-1$ rounds. A shaded region depicts a set in the partition which is fixed for the remainder of the construction, e.g. the green region is a set in $\cP_1$, and remains a set in $\cP_2,\cP_3$, etc. A white region depicts a set in the partition which may be fragmented when a later partition is defined.} 
    \label{fig:LB1}
\end{figure}

\begin{figure}
    \hspace*{4cm}
    \includegraphics[scale=0.3]{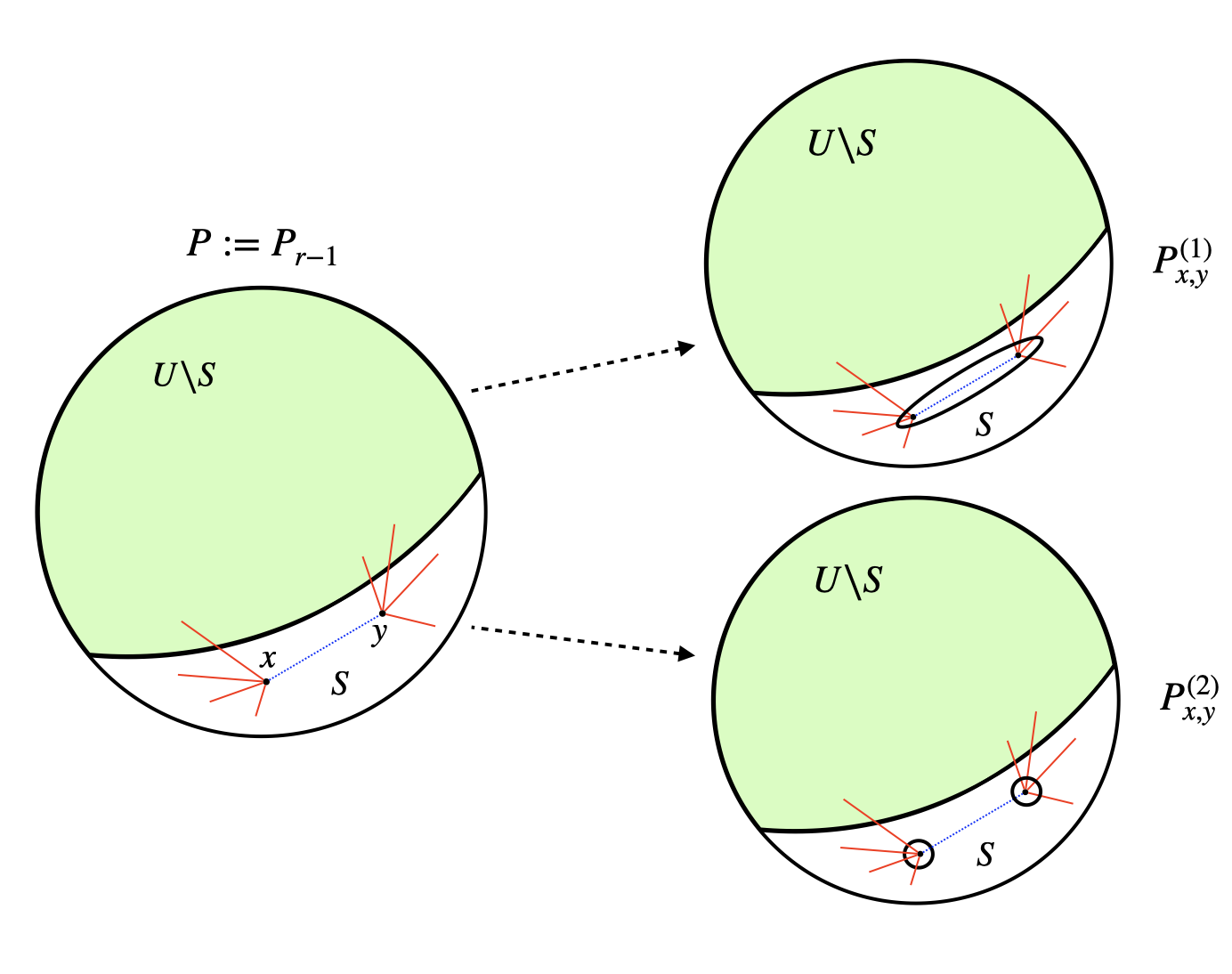}
    \caption{\small{An illustration depicting the final pair of partitions $\cP^{(1)}_{x,y}$ and $\cP^{(2)}_{x,y}$ which the $r$-round algorithm fails to distinguish.} The set $S$ is depicted in white. Note that we have not tried to depict the partition within the sets $S$ and $U \setminus S$. The point is that the partition $\cP = \cP_{r-1}$ has been defined a such a way that every queried pair $(u,v) \in Q_1 \cup \cdots \cup Q_{r-1}$ that touches $S$ is such that $u,v$ belong to different sets in $\cP$. In particular, any queried pair that touches $x$ or $y$ is given query response "no" under $\cP$, which is consistent with the partitions $\cP^{(1)}_{x,y}$ and $\cP^{(2)}_{x,y}$. Thus, these final partitions are well-defined and allow us to the lower bound the final round of queries, $Q_r$.} 
    \label{fig:LB2}
\end{figure}

\subsection{Proof of \Cref{thm:LR-pair-LB}}

For convenience throughout the proof we will use $\eps(r) := \frac{1}{2^r-1}$ for every $r \geq 1$ and $\ell := \floor{\frac{k-3}{r-1}}$ which satisfies $\ell \geq 1$ since $r \leq k-2$ by assumption.

Consider an arbitrary $r$-round deterministic algorithm and for each $t \in \{1,2,\ldots,r\}$, let $Q_t$ denote the queries made in round $t$. Note that for $t \geq 2$, this set depends on the query responses on the previous queries, $Q_1 \cup \cdots \cup Q_{t-1}$. We will assume that $|Q_1 \cup \cdots \cup Q_{r-1}| \leq \frac{n \ell}{100} (n/2\ell)^{\eps(r)}$ (since otherwise the theorem already holds for this algorithm) and using this assumption we will argue that  $|Q_r| \geq \Omega(n\ell (n/\ell)^{\eps(r)})$. 

Additionally, we will assume that $|Q_1| \geq n$ and hence $|Q_1 \cup \cdots \cup Q_t| \geq n$ for all $t \geq 1$. This is without loss of generality since increasing $|Q_1|$ can only help the algorithm. Finally, we assume that $n > C \ell$ for a sufficiently large constant $C$. %\comment{Specify this in theorem statement?}

%\comment{Handle edge case of $|Q|_1 < n$}

First, $|Q_1| \leq \frac{n \ell}{100} (n/2\ell)^{\eps(r)}$ and so we can apply the repeated Tur\'{a}n theorem (\Cref{cor:repeated-turan}) on the graph $G(U,Q_1)$ to obtain $\ell$ disjoint independent sets of size
%\comment{Need to make $N$ integral}
\[
N_1 := \left\lfloor{\left(\frac{n}{2\ell}\right) \left(\frac{2\ell}{n}\right)^{\eps(r)}} \right\rfloor \leq \frac{n^2}{2n \ell (n/2\ell)^{\eps(r)}} < \frac{n^2}{8 \cdot \frac{n \ell}{100} (n/2\ell)^{\eps(r)}} \leq \frac{n^2}{8|Q_1|} \text{.}
\]
%Clearly $N_1 < \frac{n}{2\ell}$ and so $\ell < \frac{n}{2N_1}$. Thus this choice of parameters is valid for \Cref{cor:repeated-turan} to be applied. Note also that $1-\eps(r) \geq 1-\eps(2) = 2/3$ and so $\floor{(n/2\ell)^{1-\eps(r)}} \geq (n/3\ell)^{1-\eps(r)}$ since we are assuming $n$ is sufficiently large.

%\paragraph{Construction of $\cP_1$ and $S^{(1)}$.} Let $S^{(1)}_1,\ldots,S^{(1)}_{\ell}$ be the resulting independent sets in $G(U,Q_1)$, which are each of size at least $(n/3\ell)^{1-\eps(r)}$. Let $S^{(1)}$ denote their union and note that $|S^{(1)}| = \ell N_1$.
%$|S^{(1)}| \geq \frac{n}{3} \cdot (2\ell/n)^{\eps(r)}$. 

\paragraph{Construction of $\cP_1$ and $S^{(1)}$.} Let $S^{(1)}_1,\ldots,S^{(1)}_{\ell}$ be the resulting independent sets in $G(U,Q_1)$, which are each of size at least $N_1$. Let $S^{(1)}$ denote their union and note that $|S^{(1)}| = \ell N_1$. We now define a partition $\cP_1$ which will fix the oracle's response to each query in $Q_1$ as follows:
\begin{itemize}\setlength\itemsep{0em}
    \item $U \setminus S^{(1)}$ is a set in $\cP_1$.
    \item Each $S^{(1)}_i$ is a set in $\cP_1$.
\end{itemize}

%\comment{Add picture}

All queries in $Q_1$ with both endpoints in $U \setminus S^{(1)}$ are fixed to "yes" and all others are fixed to "no" (since each $S^{(1)}_j$ is an independent set in $G(U,Q_1)$). Given these oracle responses on $Q_1$, the second round of queries $Q_2$ is now determined. Our goal is now to repeat this idea iteratively over the first $r-1$ rounds. That is, given the $t$-th round of queries $Q_t$, we wish to construct $\cP_t$ which fixes the oracle's responses on $Q_t$, and which is consistent with $\cP_{t-1}$ on all previous queries so that this process is well-defined. The following lemma is our main tool for performing a single iteration of this process. We prove this lemma in \Cref{sec:carve-proof}.

\begin{lemma} [Carving Lemma] \label{lem:carving} Suppose we have a set of queries $Q \subseteq {U \choose 2}$, $\ell \geq 1$ disjoint independent sets $S_1,\ldots,S_{\ell}$ in $G(U,Q)$, and a partition $\cP$ of $U$ such that $S_i \in \cP$ for all $i \in [\ell]$. Then, given another set of queries $Q' \subseteq {U \choose 2}$, and integers $N \leq \frac{|S|^2}{8|Q \cup Q'|}$ and $\ell' \leq |S|/2N$, there exist $\ell'$ disjoint independent sets $S_1',\ldots,S_{\ell'}'$ each of size $N$ in $G(S,Q \cup Q')$ and a partition $\cP'$ such that $|\cP'| \leq |\cP| + \ell'$ and the following hold:
\begin{enumerate}\setlength\itemsep{0em}
    \item \emph{(Inductive Property)} $S_j' \in \cP'$ for every $i \in [\ell']$.
    \item \emph{(Query Consistency)} For every $\{x,y\} \in Q$, $\same(x,y,\cP) = \same(x,y,\cP')$.
\end{enumerate}    
\end{lemma}

The idea is now to use \Cref{lem:carving} iteratively to construct a sequence of partitions $\cP_1,\cP_2,\ldots,\cP_{r-1}$ as follows. Crucially, item (2) ensures that this process is well-defined, as we will see. In what follows, for every $t \in \{1,2,\ldots,r-1\}$ define 
\begin{align} \label{eq:Nt}
    N_t := \left\lfloor \left(\frac{n}{2\ell}\right) \cdot \left(\frac{2\ell}{n}\right)^{\frac{\eps(r)}{\eps(t)}} \right\rfloor > \left(\frac{n}{3\ell}\right) \cdot \left(\frac{2\ell}{n}\right)^{\frac{\eps(r)}{\eps(t)}}
\end{align} 
where the inequality holds since $n > c\ell$ is sufficiently large and $\frac{\eps(r)}{\eps(t)} \leq \frac{\eps(r)}{\eps(r-1)} \leq 1/2$. (Note that this definition of $N_t$ is consistent with the definition of $N_1$ above, since $\eps(1) = 1$.)

\begin{lemma} \label{lem:hard-construction} For every $t \in \{1,2,\ldots,r-1\}$, there exists a set $S^{(t)} = S^{(t)}_1 \sqcup \cdots \sqcup S^{(t)}_\ell$ of size $|S^{(t)}| = \ell N_t$ and a partition $\cP_t$ of $U$ of size $|\cP_t| \leq t \ell + 1$ such that the following hold.
\begin{enumerate}\setlength\itemsep{0em}
    \item \emph{(IS Property)} For all $j \in [\ell]$, $S^{(t)}_j \in \cP_t$ and is an independent set in $G(S^{(t-1)},Q_1 \cup \cdots \cup Q_t)$.
    \item \emph{(Query Consistency)} For every $\{x,y\} \in Q_1 \cup \cdots \cup Q_{t-1}$, $\same(x,y,\cP_{t-1}) = \same(x,y,\cP_t)$.
\end{enumerate}
\end{lemma}

\begin{proof} The base case of $t=1$ is established in the paragraphs preceding the statement of \Cref{lem:carving}. I.e., we have $S^{(1)} = S^{(1)}_1 \sqcup \cdots \sqcup S^{(1)}_{\ell}$ of size $|S^{(1)}| = \ell N_1$, and we have a partition $\cP_1$ of size $|\cP_1| = \ell + 1$ such that for each $j \in [\ell]$, $S_j^{(1)} \in \cP_1$ and is an independent set in $G(U,Q_1)$.

Now let $2 \leq t \leq r-1$. Let $S^{(t-1)} = S^{(t-1)}_1 \sqcup \cdots \sqcup S^{(t-1)}_\ell$ and $\cP_{t-1}$ be the sets and partition given by induction. In particular, the query responses on $\cP_{t-1}$ now determine the $t$-th round queries $Q_t$. We now construct a partition $\cP_t$ which will define the query responses in $Q_t$. As mentioned, the responses on the previous queries $Q_1 \cup \cdots \cup Q_{t-1}$ must be consistent between $\cP_t$ and $\cP_{t-1}$ so that this process is well-defined. (Otherwise, we would contradict the definition of the sets $Q_2,\ldots,Q_{t-1}$.) See Fig. \ref{fig:LB1} in the appendix for an accompanying illustration.

%\comment{Handle edge case of $|Q_1 \cup \cdots \cup Q_t| < n$}

We will invoke \Cref{lem:carving} with $Q = Q_1 \cup \cdots \cup Q_{t-1}$ and $Q' = Q_t$. By induction $|S^{(t-1)}| \geq \frac{n}{3} (2\ell/n)^{\frac{\eps(r)}{\eps(t-1)}}$ and recall that by assumption we have $|Q_1 \cup \cdots \cup Q_t| \leq \frac{n \ell}{100}(n/2\ell)^{\eps(r)}$. Thus, using our bound on $N_{t-1}$ from \cref{eq:Nt} we have 
\begin{align} \label{eq:bound}
    \frac{|S^{(t-1)}|^2}{8|Q_1 \cup \cdots \cup Q_t|} \geq \frac{\ell^2 N_{t-1}^2}{\frac{8n\ell}{100} (n/2\ell)^{\eps(r)}}  &\geq \frac{n^2(2\ell/n)^{\frac{2\eps(r)}{\eps(t-1)}}}{\frac{72 n\ell}{100}(n/2\ell)^{\eps(r)}} \nonumber \\
    &\geq \frac{25n}{18 \ell} \cdot \left(\frac{2\ell}{n}\right)^{\frac{2\eps(r)}{\eps(t-1)}+\eps(r)} = \frac{25n}{18 \ell} \cdot \left(\frac{2\ell}{n}\right)^{\frac{\eps(r)}{\eps(t)}} > N_t
\end{align}
where the equality holds since $2(2^{t-1}-1) + 1 = 2^t - 1$. Note also that $\ell N_t < \ell N_{t-1} = |S^{(t-1)}|$. Using this observation and \cref{eq:bound}, we can invoke \Cref{lem:carving} with $N := N_t$ and $\ell' := \ell$. This yields $\ell$ disjoint independent sets $S^{(t)} := S^{(t)}_1 \sqcup \cdots \sqcup S^{(t)}_\ell$ in $G(S^{t-1}, Q_1 \cup \cdots \cup Q_t)$ and a partition $\cP_t$ such that (a) for every $j \in [\ell]$, we have $S^{(t)}_j \in \cP_t$ (by item 1 of \Cref{lem:carving}) and (b) for every $\{x,y\} \in Q_1 \cup \cdots \cup Q_{t-1}$, $\same(x,y, \cP_{t-1}) = \same(x,y, \cP_t)$ (by item 2 of \Cref{lem:carving}). Furthermore, we also have $|\cP_t| \leq |\cP_{t-1}| + \ell \leq t\ell + 1$, using the guarantee of \Cref{lem:carving} and the inductive hypothesis. This completes the proof of \Cref{lem:hard-construction}. \end{proof}

Now, invoking \Cref{lem:hard-construction} with $t := r-1$ shows that there exists a partition $\cP := \cP_{r-1}$ and a set $S := S^{(r-1)}$ such that $S = S_1 \sqcup \cdots \sqcup S_{\ell}$ and for all $j \in [\ell]$, $S_j \in \cP$ and is an independent set in $G(U,Q_1 \cup \cdots \cup Q_{r-1})$. Also note that $|\cP| \leq (r-1)\ell + 1 \leq k-2$ using the definition of $\ell$. Moreover, by \Cref{lem:hard-construction} and the bound on $N_{r-1}$ from \cref{eq:Nt}, we have
\(
|S| = \ell N_{r-1} \geq \frac{n}{3} (2\ell/n)^{\frac{2^{r-1}-1}{2^r-1}} \geq \frac{1}{3} \sqrt{n\ell (n/\ell)^{\eps(r)}}
\)
%\[
%|S| = n^{1-\frac{(2^{r-1}-1)}{2^r - 1}}k^{\frac{2^{r-1}-1}{2^r-1}} = n^{\frac{1}{2}(1+\eps(r))}k^{\frac{1}{2}(1-\eps(r))}
%\]
where the last inequality holds since
\(
\frac{2^{r-1}-1}{2^r-1} = \frac{1}{2} \cdot \frac{2^{r}-2}{2^r-1} = \frac{1}{2} \left( 1 - \frac{1}{2^r-1} \right) \text{.}
\)
This means that the number of pairs in $S$ is at least ${|S| \choose 2} \geq \frac{n\ell}{10} (n/\ell)^{\eps(r)}$. Thus, the number of un-queried pairs in $S$ is at least 
\begin{align} \label{eq:|A|}
    \left|{S \choose 2} \setminus (Q_1 \cup \cdots \cup Q_{r-1})\right| \geq \left(\frac{1}{10} - \frac{1}{100}\right)n\ell (n/\ell)^{\eps(r)} > \frac{n\ell}{12} (n/\ell)^{\eps(r)}
\end{align}
by our assumed upper bound on $|Q_1 \cup \cdots \cup Q_{r-1}|$. Let $A := {S \choose 2} \setminus (Q_1 \cup \cdots \cup Q_{r-1})$ denote this set of pairs in $S$ not queried in the first $r-1$ rounds. For every such pair $(x,y) \in A$ we define two partitions $\cP_{x,y}^{(1)}$ and $\cP_{x,y}^{(2)}$ as follows. (See Fig. \ref{fig:LB2} in the appendix for an accompanying illustration.)
\begin{itemize}\setlength\itemsep{0em}
    \item In $\cP_{x,y}^{(1)}$, $\{x,y\}$ form one set of size $2$ and in $\cP_{x,y}^{(2)}$, $\{x\}$ and $\{y\}$ each form one singleton set.
    \item In both $\cP_{x,y}^{(1)}$ and $\cP_{x,y}^{(2)}$, all other points are consistent with the partition $\cP$. Formally for every $X \in \cP$, $X \setminus \{x,y\}$ is a set in both $\cP_{x,y}^{(1)}$ and $\cP_{x,y}^{(2)}$.
\end{itemize}
Clearly, we have $|\cP_{x,y}^{(b)}| \leq |\cP| + 2 \leq k$.

First, for this to be well-defined, we need to argue that for every query $(u,v) \in Q_1 \cup \cdots \cup Q_{r-1}$, the responses given on $\cP$ and $\cP_{x,y}^{(b)}$ (for either $b \in \{1,2\}$) are consistent. Let $(u,v)$ be an arbitrary such query. Clearly if $\{u,v\} \cap \{x,y\} = \emptyset$, then $\same(u,v,\cP) = \same(u,v,\cP_{x,y}^{(b)})$. Also, note that $\{u,v\} \neq \{x,y\}$ since $(x,y) \in A$. The remaining case is when $|\{u,v\} \cap \{x,y\}| = 1$. Without loss of generality, suppose $u = x$. First, if $v \notin S$, then clearly $\same(x,v,\cP) = \same(x,v,\cP_{x,y}^{(b)}) = \mathsf{no}$ (recall that $x \in S$ by definition of the pair $(x,y)$). Now, suppose that $v \in S$. The point is that since $S = S_1 \sqcup \cdots \sqcup S_k$ where each $S_i$ is an independent set in $G(U,Q_1,\ldots,Q_{r-1})$, we must have that $x,v$ lie in different $S_i$-s, and therefore $\same(x,v,\cP) = \same(x,v,\cP_{x,y}^{(b)}) = \mathsf{no}$ (recall that each $S_i \in \cP$). Thus, all queries in the first $r-1$ rounds are consistent on $\cP$ and $\cP_{x,y}^{(b)}$. Thus, $\cP_{x,y}^{(1)}$ and $\cP_{x,y}^{(2)}$ are well-defined.

Now, we will argue that it must be the case that $Q_r \supseteq A$. Suppose not, and so there is some $(x,y) \in A$ where $(x,y) \notin Q_r$. We will argue that the algorithm cannot distinguish $\cP_{x,y}^{(1)}$ and $\cP_{x,y}^{(2)}$. Consider any pair $(u,v) \neq (x,y)$ and observe that $\same(u,v,\cP_{x,y}^{(1)}) = \same(u,v,\cP_{x,y}^{(2)})$ by construction of $\cP_{x,y}^{(1)}$ and $\cP_{x,y}^{(2)}$. Thus, since $(x,y) \notin Q_r$ and $(x,y) \notin Q_1,\ldots,Q_{r-1}$, the algorithm does not distinguish these two partitions. %This completes the proof of \Cref{thm:LB-rk}. 

This implies that $A \subseteq Q_r$ and consequently $|Q_r| \geq \frac{n\ell}{12} (n/\ell)^{\eps(r)}$ by \cref{eq:|A|}. This completes the proof of \Cref{thm:LR-pair-LB}. 
%\end{proofof} %\end{proof}

%\begin{lemma} There exists a sequence of partitions $\cP_1,\ldots,\cP_{r-1}$ and a sequence of sets $S^{(1)} \supseteq \cdots \supseteq S^{(r-1)}$ such that the following are satisfied for every $t \in \{1,\ldots,r-1\}$.
%\begin{enumerate} 
%    \item Every $S^{(t)} = S^{(t)}_1 \sqcup \cdots \sqcup S^{(t)}_{k}$ is a disjoint union of independent sets in $G(S^{(t-1)},Q_1 \cup \cdots \cup Q_t)$.
%    \item $|S^{(t)}| = n^{1-(2^t-1)\eps(r)}k^{(2^t-1)\eps(r)}$.

%    \item The total number of sets in $\cP_t$ is at most $tk+1$.
%\end{enumerate}    
%\end{lemma}

%\begin{proof} The base case of $t=1$ is satisfied by the construction in the previous paragraphs. Now let $2 \leq t \leq r-1$. Let $\cP_{t-1}$ $S^{(t-1)}$ denote the construction given by the inductive hypothesis.

%which establishes item (1). We now define $\cP_t$ using $\cP_{t-1}$ and the independent sets:
%\begin{itemize}
%    \item We make $S^{(t)}_1,\ldots,S^{(t)}_k$ each a set in $\cP_t$.
%    \item For each $X \in \cP_{t-1}$, we make $X \setminus S^{(t)}$ a set in $\cP_t$. 
%\end{itemize}
%The number of sets in the partition has increased by a most $k$, establishing item (3). We now need to prove consistency, i.e. item (2) of the lemma. Consider any $\{x,y\} \in Q_1 \cup \cdots \cup Q_{t-1}$. 
%\end{proof}

\subsection{Deferred Proofs} \label{sec:deferred-proofs}

\subsubsection{Tur\'{a}n's Theorem} \label{sec:Turan-proof}

\begin{proofof}{\Cref{thm:turan}} Pick a random ordering $\pi$ on the vertices and construct an independent set $I_{\pi}$ greedily according to $\pi$. I.e., iterating over $i \in [n]$, the $\pi(i)$-th vertex is added to $I_{\pi}$ iff it has no neighbors in $I_{\pi}$. Then, $v$ appears in $I_{\pi}$ if it comes before all of its neighbors in $\pi$, which occurs with probability exactly $1/(1+d_G(v))$. Thus,
\begin{align}
    \mathbb{E}_{\pi}[|I_\pi|] = \sum_{v \in V} \frac{1}{1+d_G(v)} = n \cdot \mathbb{E}_{v \in V}\left[\frac{1}{1+d_G(v)}\right] \text{.}
\end{align}
The function $1/(1+x)$ for $x \geq 0$ is convex and so by Jensen's inequality
\[
\mathbb{E}_{v \in V}\left[\frac{1}{1+d_G(v)}\right] \geq \frac{1}{1+\mathbb{E}_{v \in V}[d_G(v)]} = \frac{1}{1+d_G}
\]
and this completes the proof. \end{proofof}

\subsubsection{Repeated Tur\'{a}n's Theorem} \label{sec:repeated-Turan-proof}

\begin{proofof}{\Cref{cor:repeated-turan}} The average degree in $G$ is $2m/n$, and so by Tur\'{a}n's \Cref{thm:turan}, it contains an independent set of size at least $\frac{n}{(2m/n)+1} \geq \frac{n^2}{4m}$, where here we have used $m \geq n$. Thus, we can take $S_1$ to be an independent set of size $N \leq \frac{n^2}{8m}$. Now, suppose we have constructed $1 \leq t < \ell$ disjoint independent sets $S_1,\ldots,S_t$, each of size $N$. Let $G_t = G[U \setminus (S_1 \cup \cdots \cup S_t)]$ denote the induced subgraph obtained by removing these independent sets. The number of edges in $G_t$ is clearly still at most $m$ and the number of vertices is at least 
\[
n - t N \geq n - \ell N \geq n/2
\]
since $t < \ell \leq n/2N$. Thus, the average degree in $G_t$ is at most $2m/(n/2) \leq 4m/n$, and again by Tur\'{a}n's \Cref{thm:turan} we get an independent set in $G_t$ of size at least $\frac{n}{(4m/n)+1} \geq \frac{n^2}{8m}$ and in particular we can take $S_{t+1}$ to be an independent set of size $N$. This completes the proof. \end{proofof}

\subsubsection{The Carving Lemma} \label{sec:carve-proof}

%\comment{Add picture}

\begin{proofof}{\Cref{lem:carving}} First, we apply the repeated Tur\'{a}n Theorem (\Cref{cor:repeated-turan}) in $G(S,Q \cup Q')$ to obtain the independent sets $S_1',\ldots,S_{\ell'}'$ each of size $N$ and let $S'$ denote their union. Note that this is possible by the assumed bounds on $N$ and $\ell'$. We now define the partition $\cP'$ as follows: (a) each $S_i' \in \cP'$ is a set of the partition, and (b) for each $X \in \cP$, we make $X \setminus S' \in \cP'$ a set in the partition. Note that $\cP'$ clearly is a partition of $U$, $|\cP'| \leq |\cP| + \ell'$, and item (1) holds by construction.

We now prove that item (2) holds. Consider any $\{x,y\} \in Q$. We break into cases depending on where $x,y$ lie. Note that by our assumption about $S$ and $\cP$ and the construction of $\cP'$, for every set $X \in \cP \setminus \{S_1,\ldots,S_{\ell}\}$, we also have $X \in \cP'$, i.e. these sets are preserved. Thus, if $x,y \in U \setminus S$, then clearly $\same(x,y,\cP) = \same(x,y,\cP')$. This also implies that if exactly one of $x$ or $y$ lie in $U \setminus S$, then $\same(x,y,\cP)] = \same(x,y,\cP') = \mathsf{no}$. 

The remaining case to consider is when both $x,y \in S$. Now recall that $S = S_1 \sqcup \cdots \sqcup S_{\ell}$ where each $S_i$ is an independent set in $G(U,Q)$. In particular, this means that $x,y$ lie in \emph{different} $S_i$-s. Let $x \in S_i, y \in S_j$ where $i \neq j$. This means that $\same(x,y,\cP) = \mathsf{no}$. Now, note that $S_i \setminus S' \in \cP'$ and $S_j \setminus S' \in \cP'$ where (a) $S' = S_1' \sqcup \cdots \sqcup S_{\ell'}'$, (b) each $S_i'$ is an independent set in $G(S,Q \cup Q')$, and (c) each $S_i' \in \cP'$ is a set in $\cP'$. In particular, if both, or exactly one, of $x,y$ lie in $S \setminus S'$, then we clearly have $\same(x,y,\cP') = \mathsf{no}$. Finally, if both $x,y \in S'$, then by (b) we must have that $x,y$ are in different $S_i'$-s, implying that $\same(x,y,\cP')] = \mathsf{no}$. This completes the proof. \end{proofof}

%\subsection{Proof of Lemma~\ref{lem:hard-construction}}

\section{Low-Round Algorithm using Pair Queries} \label{sec:LR-pair}

In this section we prove \Cref{thm:LR-pair-UB}, obtaining an $r$-round deterministic algorithm for learning a $k$-partition which interpolates between the $\Theta(n^2)$ non-adaptive query complexity and the $\Theta(nk)$ fully adaptive query complexity. %In particular, our algorithm only needs $O(\log \log n)$ rounds to match the optimal query complexity $O(nk)$. 
We use a simple recursive strategy described in pseudocode in \Alg{low-round-pair}: divide $U$ into subproblems (line 6), compute the restricted partition in each part by (non-adaptively) querying every pair (line 8), and then recurse on a set $R$ formed by taking exactly one representative from each set in each of the restricted partitions (lines 9-11). For the base case, simply use the trivial strategy of querying all pairs. \\

\begin{proofof}{\Cref{thm:LR-pair-UB}} For shorthand in the proof, we will use $\eps(r) = \frac{1}{2^{r}-1}$ for all $r \geq 1$. The algorithm is recursive and pseudocode is given in \Alg{low-round-pair}. We prove the theorem by induction on $r$. For the base case of $r=1$ (lines (2-4)), we simply make all ${n \choose 2}$ possible pairwise queries in $U$. The partition can trivially be recovered using this set of queries. Moreover, the number of queries is at most the desired bound since $\eps(1) = 1$.

\begin{algorithm}[ht]
\caption{$\LRPQ(U,r)$\label{alg:low-round-pair}} 
\textbf{Input:} Pair query access to hidden partition $\cC$ over $U$ with $n$ points and $|\cC| \leq k$. An allowed number of rounds $r$\; 
\If{$r = 1$ \textbf{or} $n \leq 16k$}{ 
    \textbf{Query} every pair $(x,y) \in {U \choose 2}$ and \textbf{return} the computed partition\; 
} \Else{

Partition the $n$ points of $U$ arbitrarily into at most $\ell := \ceil{\frac{1}{3}\left(\frac{n}{k}\right)^{1-\frac{1}{2^r-1}}}$ sets $U_1,\ldots,U_{\ell}$ each of size $|U_i| \leq t := \ceil{3n^{\frac{1}{2^r-1}} k^{1-\frac{1}{2^r-1}}}$\;

\For{$i\in [\ell]$} {
\textbf{Query} every pair $(x,y) \in {U_i \choose 2}$ to learn the partition $\cC_i = \{C \cap U_i \colon C \in \cC \}$\;
Form $R_i$ by taking exactly one representative from each set $C \in \cC_i$\;
}
Let $R = R_1 \cup \cdots \cup R_{\ell}$, call $\LRPQ(R,r-1)$, and let $\cC_R$ be the returned partition of $R$\;
\textbf{Return} the partition $\{\bigcup_{C' \in \cC_1 \cup \cdots \cup \cC_{\ell} \colon C' \cap C \neq \emptyset} C' \colon C \in \cC_R\}$\;
}
\end{algorithm}

Now suppose $r \geq 2$. First, if $n \leq 16k$, then we simply make all pair-wise queries in $U$ (line 3), for a total of ${n \choose 2} \leq \frac{n^2}{2} \leq 8nk \leq 8n^{1+\eps(r)}k^{1-\eps(r)}$ queries, where the last step simply used $n \geq k$. Note that in this case the algorithm is non-adaptive and the partition is trivially recovered. This completes the proof for the case of $n \leq 16k$.

Now suppose that $n > 16k$. Let us first establish correctness of the algorithm. First, we take a partition $U = U_1 \sqcup \cdots \sqcup U_{\ell}$ and then within each $U_i$ we make all pair-wise queries and trivially recover the partition restricted in $U_i$, i.e. $\cC_i = \{C \cap U_i \colon C \in \cC\}$. Next, our goal is to merge these partitions to recover $\cC$. To do so, we form a set $R$ containing one representative from every set in $\cC_i$ for every $i \in [\ell]$ (lines 9 and 11), and recursively learn the partition restricted on $R$, i.e. $\cC_R = \{C \cap R \colon C \in \cC\}$ (line 11). By induction this correctly computes $\cC_R$ and allows us to compute the final partition by merging all sets $C' \in \cC_1,\ldots,\cC_{\ell}$ which intersect the same set $C \in \cC_R$. This completes the proof of correctness.

We now complete the proof of the claimed query complexity. Recall that we are in the case of $r \geq 2$ and $n > 16k$. Since $r \geq 2$, note that $1-\eps(r) \geq 1-\eps(2) \geq 2/3$. Using these bounds, we have $(n/k)^{1-\eps(r)} > 16^{2/3} > 6$. Recalling the definition of $\ell$ and $t$ in line (6), this implies that
\(
\ell \leq \frac{1}{3}\left(\frac{n}{k}\right)^{1-\eps(r)} + 1 = \frac{1}{2}\left(\frac{n}{k}\right)^{1-\eps(r)} - \frac{1}{6}\left(\frac{n}{k}\right)^{1-\eps(r)} + 1 < \frac{1}{2}\left(\frac{n}{k}\right)^{1-\eps(r)}
\)
and clearly $t \leq 4n^{\eps(r)}k^{1-\eps(r)}$. Thus, the first round (line 8) makes at most
\begin{align} \label{eq:round_1}
     \ell \cdot {t \choose 2} < \frac{1}{2}\left(\frac{n}{k}\right)^{1-\eps(r)} \cdot 8n^{2\eps(r)} k^{2(1-\eps(r))} \leq 4n^{1+\eps(r)} k^{1-\eps(r)}
\end{align}
queries. Then, the resulting set $R$ defined in line (11) is of size
\(
|R| \leq k \cdot \frac{1}{2} \left(\frac{n}{k}\right)^{1-\eps(r)} = \frac{1}{2} \cdot n^{1-\eps(r)} k^{\eps(r)}\text{.}
\)
By induction, the recursive call to $\LRPQ(R,r-1)$ in line (11) then costs 
\begin{align} \label{eq:round_r-1}
    8\left(\frac{1}{2} \cdot n^{1-\eps(r)} k^{\eps(r)}\right)^{1+\eps(r-1)} k^{1-\eps(r-1)} \leq 4n^{(1-\eps(r))(1+\eps(r-1))} k^{\eps(r)(1+\eps(r-1))+(1-\eps(r-1))}
\end{align}
queries at most. To understand the exponents in the RHS we need the following claim about $\eps(r)$. 

\begin{claim} \label{clm:eps(r)} For all $r\geq 1$, let $\eps(r) = \frac{1}{2^r-1}$. The following hold for all $r \geq 2$:
\begin{enumerate}
    \item $(1-\eps(r))(1+\eps(r-1)) = 1+\eps(r)$.
    \item $\eps(r)(1+\eps(r-1))+(1-\eps(r-1)) = 1-\eps(r)$.
\end{enumerate}
\end{claim}

\begin{proof} To see that item (1) holds, observe that 
\begin{align*}
    (1-\eps(r))(1+\eps(r-1)) = \frac{2^r-2}{2^r-1} \cdot \frac{2^{r-1}}{2^{r-1}-1} = 2 \cdot \frac{2^{r-1}-1}{2^r-1} \cdot \frac{2^{r-1}}{2^{r-1}-1} = \frac{2^r}{2^r -1} = 1+\eps(r)\text{.}
\end{align*}
To see that item (2) holds, observe that the statement is equivalent to the identity $\eps(r) = \frac{\eps(r-1)}{2 + \eps(r-1)}$. This identity holds since 
\begin{align}
    \frac{\eps(r-1)}{2 + \eps(r-1)} = \frac{1}{(2^{r-1}-1)(2 + \frac{1}{2^{r-1}-1})} = \frac{1}{2(2^{r-1}-1)+1} = \frac{1}{2^r - 1} = \eps(r) 
\end{align}
and this completes the proof. \end{proof}    

Now, by \Cref{clm:eps(r)} the RHS of \cref{eq:round_r-1} is equal to $4n^{1+\eps(r)}k^{1-\eps(r)}$. Combining this with the bound \cref{eq:round_1} on the number of queries in the first round shows that $\LRPQ$ makes at most $8n^{1+\eps(r)}k^{1-\eps(r)}$ queries. This completes the proof of \Cref{thm:LR-pair-UB}. \end{proofof}

\section{Weak Subset Queries} \label{sec:count}

%In this section we design and analyze our non-adaptive subset query algorithms. 
In this section we provide a nearly-optimal randomized non-adaptive algorithm using weak 
%(and strong, see  Theorem \ref{thm:NA-strong}) 
subset queries. We then design a nearly-optimal $r$-round algorithm  following the recursive algorithmic template for $r$-rounds described in \Alg{low-round-pair}, with weak and strong subset queries (\Cref{sec:LR-count} and \Cref{sec:LR-partition}), where the trivial all-pair-query subroutine is replaced by the best non-adaptive algorithm for the respective query type. %\smallskip

%We use this same algorithmic template to obtain nearly optimal query and round complexity for learning partitions 

%\subsection{Nearly Optimal Non-adaptive Algorithm} \label{sec:NA-count}

%This section is dedicated to proving the following theorem.

\begin{theorem}[Weak Subset Query Non-adaptive Algorithm] \label{thm:NA-bounded} There is a non-adaptive algorithm which, for any query size bound $2 \leq s \leq \sqrt{n}$ and error probability $\delta > 0$, learns an arbitrary partition on $n$ elements exactly using 
\[
O\left(\frac{n^2}{s^2} \log (n/\delta) + n \log^4 (n/\delta) \log s \right) 
\]
weak subset queries of size at most $s$, and succeeds with probability $1-\delta$. Observe that if $s \leq O(\frac{\sqrt{n}}{\log^2(n/\delta)})$, then the query complexity becomes $O(\frac{n^2}{s^2} \log (n/\delta))$. %If every set of the hidden partition is known to have size $\Omega(n/s^2)$, then the algorithm makes $O(n \log^4 (n/\delta) \log s)$ queries. 
\end{theorem}

We provide an overview of our proof in \Cref{sec:NA-subset-overview} and the full detailed proof in \Cref{sec:NA-bounded-proof}. Pseudocode for the algorithm is given in \Alg{opt-bounded}.

%First, note that for sufficiently small $s$ the second term in the sum above dominates and so the $\log$-factors improve significantly. In particular, we note the following corollary.

%\begin{corollary} If $2 \leq s \leq O(\frac{\sqrt{n}}{\log^2(n/\delta)})$, then there is a non-adaptive algorithm for learning an arbitrary partition making $O(\frac{n^2}{s^2} \log (n/\delta))$ weak subset queries of size at most $s$. \end{corollary}

%\comment{Need to do a pass over this proof to polish it up. Maybe $\log$-factors can be improved.}

% \begin{theorem}[Strong Subset Query Non-adaptive Algorithm] \label{thm:NA-strong} For any even $2 \leq s \leq n$, there is a deterministic non-adaptive algorithm for learning an arbitrary partition using at most $\frac{5n^2}{s^2}$ strong subset queries of size $s$. \end{theorem}

\subsection{Proof Overview of \Cref{thm:NA-bounded}} \label{sec:NA-subset-overview}

The algorithm proceeds by iteratively learning the sets of the partition from largest to smallest. This is divided into three phases in which we learn the "large" sets (size at least $n/\poly\log(n/\delta)$, lines 4-9), then the "medium" sets (size at most $n/\poly\log(n/\delta)$ and at least $n/s^2$, lines 10-22), and finally the "small" sets (size at most $n/s^2$, lines 23-25). The algorithm maintains $\tC$ containing the sets learned so far. 
To learn the medium and small sets it always exploits the known larger sets in $\tC$ to perform the reconstruction. Note however that the queries made by the algorithm never depend on $\tC$. In particular, all queries can be made in advance before any reconstruction is performed and so the algorithm is indeed non-adaptive.\footnote{One could alternatively present the pseudocode of \Alg{opt-bounded} so that all queries are made first before any reconstruction, making the non-adaptivity of the algorithm more transparent. However, we believe that presenting the query-selection and reconstruction process together makes the pseudocode much more intuitive and readable.}

The large sets are easiest to learn: sample a random set $R$ which is large enough to contain a representative from every large set with high probability, and then make all pairwise queries between $R$ and $U$. The point is that one only needs $|R| \approx \poly\log(n/\delta)$ for this to hold. To learn the medium and small sets we exploit a subroutine $\LS$ (see \Alg{sparse}) for learning partially reconstructed partitions (the subroutine is provided a collection of known sets) where every unknown set is sufficiently small. This procedure (see \Cref{lem:sparse}) learns all unknown sets with high probability when each unknown set has size at most $c$ using only $\widetilde{O}(cn)$ queries of size $\sqrt{n/c}$. Note this allows us to learn the "small" sets (of size at most $n/s^2$) using $\widetilde{O}(n^2/s^2)$ queries of size at most $s$ as desired. %(assuming that the large and medium sets have been reconstructed successfully). %Indeed this is all that needs to be done to learn the small sets. The most interesting and challenging case is the medium-sized sets.

\paragraph{Learning the medium-sized sets.} The medium sets are too large to apply the subroutine $\LS$ directly and obtain the desired query complexity $\widetilde{O}(n^2/s^2)$. To circumvent this, the key idea is to sample smaller subsets $X \subset U$ and learn the partition restricted on each subset using $\LS$, and then piece these solutions together. This is done in the body of the while-loop (lines 12-21) during which the iteration corresponding to value $B$ tries to learn the remaining unknown sets of size at least $n/B$. %where at the end of the loop body $B$ is multiplied by $2$ and so there are only $O(\log(s))$ iterations. 
To learn these sets, we sample a random "core" set $R$ of size $\widetilde{O}(B)$ (line 14) so that with high probability $R$ contains a representative from every unknown set of size at least $n/B$. Then, we sample $p = \widetilde{O}(n/B)$ random sets $X_1,\ldots,X_p$ of size $B$ (line 15). With high probability the following will hold: (i) the $X_j$-s cover $U$, (ii) $R$ contains a representative from every unknown set of size at least $n/B$, and (iii) for every unknown set $C$ (all of these are of size $\leq 2n/B$), we have $|C \cap (X_j \cup R)| \leq O(\log(n/\delta))$. By (i-ii) it suffices to learn the partition restricted on each $X_j \cup R$, and by (iii) this can be done efficiently using $\LS$ (\Cref{lem:sparse}). (See \Fig{sunflower} for an accompanying illustration.)

\begin{algorithm}[H]
\caption{$\NASQ(n,k,s,\delta)$\label{alg:opt-bounded}} 
\textbf{Input:} Subset query access to a hidden partition $\cC$ of $|U| = n$ points into at most $k$ sets. An error probability $\delta > 0$\;

\textbf{Output:} A partition $\widetilde{\cC}$ of $U$ which is equal to $\cC$ with probability $1-\delta$\;

Initialize hypothesis partition $\widetilde{\cC} \gets \emptyset$\;

\emph{(Learn the large sets)}\;

Sample a set $R \subset U$ of $\ln^2(n/\delta) \ln(k/\delta)$ i.i.d. uniform random elements\;

\For{$x \in R$, $y \in U$} {
\textbf{Query} $\{x,y\}$ and let $C_x = \{y \in U \colon \countq(\{x,y\}) = 1\}$ denote the set containing $x$\;

\textbf{If} $|C_x| \geq \frac{n}{\ln^2(n/\delta)}$, then $\widetilde{\cC} \gets \widetilde{\cC} \cup \{C_x\}$\;
}

$\backslash \backslash$ $\triangleright$ $\widetilde{\cC}$ now contains all sets in the partition of size at least $\frac{n}{\ln^2(n/\delta)}$, with probability $1-\delta$. $\backslash \backslash$

\emph{(Learn the medium sets)}\;
$B \gets 2\ln^2(n/\delta))$\;
\While{$B \leq s^2$}{
\emph{(Learn the unknown sets of size at least  $\frac{n}{B}$)}\;
Sample a set $R \subset U$ of $B \ln (6 B\log(s^2)/\delta)$ i.i.d. uniform random elements\;
Sample $p = \frac{n}{B} \ln (6n \log (s^2)/\delta)$ sets $X_1,\ldots,X_p$ each of $B$ i.i.d. uniform random elements\;
\For{$j \in [p]$}
{
%Let $c \gets 30 \ln (n/\delta)$\;
Let $\cK_j = \{C \cap (X_j \cup R) \colon C \in \widetilde{C}\}$ denote the current known sets in partition restricted on $X_j \cup R$\; 
Run $\LS(\cK_j, 72 \ln (n/\delta),\delta/(2p \log(s^2)))$ on $X_j \cup R$ and let $G_j$ be the returned partition-graph on $(X_j \cup R) \setminus \bigcup_{K \in \cK_j} K$\;
$\backslash \backslash$ $\triangleright$ Note that the queries made by $\LS$ do not depend on the set $\widetilde{\cC}$. $\backslash \backslash$ 
}
Let $\widetilde{C}_1,\ldots,\widetilde{C}_{\ell}$ denote the connected components of the union $G = G_1 \cup \cdots \cup G_p$\;
Update $\widetilde{\cC} \gets \widetilde{\cC} \cup \{\widetilde{C}_j \colon j \in [\ell] \text{, } |\widetilde{C}_{j}| \geq n/B\}$ and $B \gets 2B$\;
}
$\backslash \backslash$ $\triangleright$ $\widetilde{\cC}$ now contains all sets in the partition of size at least $\frac{2n}{s^2}$ with probability $1-2\delta$. $\backslash \backslash$

\emph{(Learn the small sets)}\;
Run $\LS(\tC, 2n/s^2, \delta)$ and let $G$ denote the returned partition-graph on $U \setminus \bigcup_{C \in \widetilde{\cC}} C$\;

$\backslash \backslash$ $\triangleright$ Note that the queries made by $\LS$ do not depend on the set $\widetilde{\cC}$. $\backslash \backslash$

Add the connected components of $G$ to $\widetilde{\cC}$\;

$\backslash \backslash$ $\triangleright$ At this stage $\widetilde{\cC}$ is exactly equal to $\cC$ with probability $1-3\delta$. $\backslash \backslash$

\textbf{Return} $\widetilde{\cC}$\;
\end{algorithm}

\paragraph{The subroutine for learning sparse partitions.} We now describe the subroutine $\LS$, which is described in pseudocode in \Alg{sparse}. The idea is to learn all pairs of elements which belong to \emph{different} sets in the partition. 
%If all such pairs are learned, then the partition can be recovered. 
We say that a set $I$ is an \emph{independent set} if every element of $I$ belongs to a different set in the partition. Upon querying an independent set $I$, the oracle responds with $\countq(I) = |I|$, and we learn that all pairs in $I$ belong to different sets, i.e. we learn the relationship of $\approx |I|^2$ pairs. If all sets in the partition are of size at most $c$, then a random $I$ of size $\approx \sqrt{n/c}$ will be an independent set with constant probability, allowing us to learn essentially $\approx n/c$ pairwise relationships per query, leading to a $\approx c n$ query algorithm, since there are at most $n^2$ pairs to learn in total. 

However, this is not the whole story: recall that the only guarantee of \Cref{lem:sparse} is that the \emph{unknown} sets are of size at most $c$, while the \emph{known} sets can be arbitrarily large. Let $K$ denote the union of the known sets. Then, we just need to learn every pair of elements $u,v \in U \setminus K$ which belong to different unknown sets. The point is that since $K$ is known, the oracle response to $I \setminus K$ can be simulated using the query to $I$ (see line 7 and \Fig{IS-sim} for an illustration). Therefore, it suffices to query $I$ and this query does not depend on $K$ at all. Thus, the queries selected by the subroutine can be made without knowledge of $K$, allowing the main algorithm (\Alg{opt-bounded}) to be implemented non-adaptively.

\subsection{Non-adaptive Weak Subset Query Algorithm: Proof of \Cref{thm:NA-bounded}} \label{sec:NA-bounded-proof}

Let $\cC = (C_1,\ldots,C_k)$ denote the hidden partition on $U$. For each $B \in [1,n]$, let $\cC_B = \{C \in \cC \colon |C| \geq n/B\}$ denote the collection of sets in $\cC$ of size at least $n/B$. Our main subroutine will be from the following lemma, which we prove in \Cref{sec:sparse}.

\begin{lemma} [Subroutine for learning a sparse, partially known partition] \label{lem:sparse} Let $\cC$ be a hidden partition over a universe $U$ on $n$ points. Suppose that a subset of the partition $\cK \subseteq \cC$ is completely known and let $K = \bigcup_{C \in \cK} C$. There is a non-adaptive procedure with subset query access to $\cC$, $\LS(\cK,c,\delta)$ which makes $2 c n \ln(n^2/\delta)$ queries of size $\floor{\sqrt{n/c}}$ and returns a graph $G$ on $U \setminus K$ whose connected components are exactly the set of unknown sets $\cC \setminus \cK$ with probability $1-\delta$, if every unknown set $C \in \cC \setminus \cK$ satisfies $|C| \leq c$. Moreover, the queries made by the procedure do not depend on the known sets, $\cK$. \end{lemma} 

Pseudocode for our algorithm is given in \Alg{opt-bounded}. The algorithm works by maintaining a set $\tC$, which at any point in the algorithm's execution will be equal to $\cC_B$ for some $B \in [1,n]$ with high probability, as we will show. In particular, when the algorithm terminates, we will have $\tC = \cC_{n}$ with probability $1-\delta$. In lines (6-9) we employ a simple strategy to learn all of the sufficiently large sets.

\begin{claim} \label{clm:large} After line (9) of \Alg{opt-bounded}, we have $\tC = \cC_{n/\ln^2 (n/\delta)}$ with probability $1-\delta$. Moreover, the number of queries made in line (7) is at most $n \ln^2 (n/\delta) \ln(k/\delta)$ and every query is of size $2$. \end{claim}

\begin{proof} The number of queries made in line (7) is exactly $|U| \cdot |R| = n \ln^2 (n/\delta) \ln(k/\delta)$, as claimed. Observe that in line (7), the set $C_x$ is exactly the set in $\cC$ which contains the point $x$. Thus, the claim holds as long as $R \cap C \neq \emptyset$ for every $C \in \cC_{n/\ln^2 (n/\delta)}$. For such a fixed $C$, we have 
\[
\Pr_R[R \cap C = \emptyset] = \left(1-\frac{|C|}{n}\right)^{|R|} \leq \left(1-\frac{1}{\ln^2(n/\delta)}\right)^{\ln^2 (n/\delta) \ln(k/\delta)} \leq \delta/k
\]
and so the claim holds by a union bound over the at most $k$ sets in $\cC_{n/\ln^2 (n/\delta)}$. \end{proof}

Next, we employ a different strategy in lines (13-23) to learn the sets with sizes in $[\frac{2n}{s^2},\frac{n}{\ln^2(n/\delta)}]$. This is the most involved phase of the algorithm.

\begin{lemma} \label{lem:medium} If $\tC = \cC_{n/\ln^2(n/\delta)}$ in line (9) of \Alg{opt-bounded}, then in line (23), we have $\tC = \cC_{2n/s^2}$ with probability $1-\delta$. Moreover, the number of queries made by $\LS$ in line (19) is $O(n \ln^4 (n/\delta) \ln s)$ and every query is of size $s$. \end{lemma}

\begin{proof} Note that if $s^2 \leq 2\ln^2(n/\delta)$, then the lemma is vacuously correct. We will prove the lemma by an induction on each setting of $B$ during the while loop in lines (13-23), where $2\ln^2 (n/\delta) \leq B \leq s^2$ and $B$ doubles after each iteration, in line (21). The following claim proves correctness and the desired query complexity for each iteration, and \Cref{lem:medium} then follows immediately since there are at most $\log(s^2)$ iterations in total.

\begin{claim} Consider an arbitrary iteration of the while loop beginning in line (13). If in line (13) it holds that $\tC = \cC_{B/2}$, then in line (22), we have $\tC = \cC_B$ with probability $1-\frac{\delta}{\log (s^2)}$. Moreover, during this iteration the number of queries made by the calls to $\LS$ in line (19) is $O(n \ln^4 (n/\delta))$, and every query is of size at most $s$. \end{claim}

%\comment{Add picture of what's going on the while loop. Picture of a sunflower where petals correspond to $X_j$'s and core corresponds to $R$. Sparse learner runs on each petal}

\begin{proof} First, we define the following good events. We will argue that conditioned on these events, we will have $\widetilde{\cC} = \cC_B$ with probability $1-\frac{\delta}{2\log(s^2)}$. We will then argue that with probability $1-\frac{\delta}{2\log(s^2)}$ all the events occur, and this will prove the claim by a union bound.

\begin{itemize}
    \item Let $\cE_{R,\mathsf{cover}}$ denote the event that $R \cap C \neq \emptyset$ for every $C \in \cC_B$.
    \item Let $\cE_{X,\mathsf{cover}}$ denote the event that $X_1 \cup \cdots \cup X_p = U$.
    \item Let $\cE_{\mathsf{sparse}}$ denote the event that $|(X_j \cup R) \cap C| \leq 72 \ln(n/\delta)$ for every $C \in \cC \setminus \cC_{B/2}$ and every $j \in [p]$.
\end{itemize}

Recall we are conditioning on $\tC = \cC_{B/2}$ and so in line (18) $\cK_j$ is exactly this collection of sets in the partition restricted on the set $X_j \cup R$. If the call to $\LS$ on $X_j \cup R$ in line (19) succeeds, then the connected components of the returned graph $G_j$ are exactly $X_j \cap C$ for each $C \in \cC_B$. Moreover, if $\cE_{R,\mathsf{cover}}$ occurs, then the connected components of the union $G = G_1 \cup \cdots \cup G_p$ are exactly $(X_1 \cup \cdots \cup X_p) \cap C$ for each $C \in \cC_B$. This shows that if $\cE_{X,\mathsf{cover}}$ occurs, then the set of connected components of $G$ is exactly $\cC_B$ and we get $\tC = \cC_B$ as desired in line (22). Finally, conditioned on $\cE_{\mathsf{sparse}}$, \Cref{lem:sparse} guarantees that each call to $\LS$ in line (18) succeeds with probability $1-\delta/(2p\log(s^2))$, and so by a union bound over the $p$ calls, all of them succeed with probability $1-\delta/2 \log(s^2)$ and we have $\tC = \cC_B$ as argued.

Now, to complete the proof it suffices to show that each good event $\cE_{R,\mathsf{cover}}$, $\cE_{X,\mathsf{cover}}$, and $\cE_{\mathsf{sparse}}$ occurs with probability at least $1-\delta/6 \log(s^2)$.

%First, for the set $R$ drawn in line (15), let $\cE_R$ denote the good event that $R \cap C \neq \emptyset$ for every $C \in \cC_B$. For the $X_j$-s drawn in line (16), let $\cE_X$ denote the good event that $X_1 \cup \cdots \cup X_p = U$. 

Recall from line (14) that $R$ is a set of $B \ln (6 B\log(s^2)/\delta)$ i.i.d. uniform random elements. Thus, observe that for $C \in \cC_B$, 
\begin{align} \label{eq:R}
\Pr_R[R \cap C = \emptyset] \leq \left(1-\frac{1}{B}\right)^{|R|} \leq \delta/(6B \log (s^2)) ~\Longrightarrow~ \Pr_R[\cE_{R,\mathsf{cover}}] \geq 1-\delta/(6\log(s^2))
\end{align}
by a union bound since there can be at most $B$ sets in $\cC_B$. 

Now, recalling the definition of $p$ and $X_1,\ldots,X_p$ in line (15), observe that $X = X_1 \cup \cdots \cup X_p$ is simply a set of $n \ln (6n\log (s^2)/\delta)$ i.i.d. uniform random samples from $U$. Thus, the probability that $x \notin X$ for a fixed $x \in U$ is at most $(1-1/n)^{|X|} \leq \delta/(6n\log(s^2))$ and so by a union bound over all $x \in U$, we have
\begin{align} \label{eq:X}
    \Pr_X[\cE_{X,\mathsf{cover}}] \geq 1-\delta/(6\log(s^2)) \text{.}
\end{align}
We now lower bound the probability of $\cE_{\mathsf{sparse}}$. First, $X_j \cup R$ is simply a set of 
\begin{align} \label{eq:RcupX}
	|R \cup X_j| \leq |R| + |X_j| \leq 2B\ln(6B \log (s^2)/\delta) \leq 4 B \ln(n/\delta)
\end{align}
independent uniform random elements\footnote{This inequality holds since $B \leq s^2 \leq n$ and so the argument of the $\log$ is $6B \log(s^2)/\delta \leq n^2/\delta < (n/\delta)^2$.}. Fix a set $C \in \cC \setminus \cC_{B/2}$, i.e. $C$ is of size at most $2n/B$. Thus, a uniform random element lands in $C$ with probability at most $2/B$, and so 
\[
\Pr_{R,X_j}[|(X_j \cup R) \cap C| > t] \leq {4 B \ln(n/\delta) \choose t}\left(\frac{2}{B}\right)^t \leq \left(\frac{4 e B \ln(n/\delta)}{t} \cdot \frac{2}{B}\right)^t \leq \left(\frac{24 \ln (n/\delta)}{t}\right)^t
\]
where the first inequality is by taking a union bound over each subset of $t$ random elements from $X_j \cup R$ and considering the probability that they all land in $C$. The second inequality follows from the inequality ${n \choose t} \leq (\frac{en}{k})^k$. Thus, if we set $t =  24 e \ln (n/\delta) < 72 \ln (n/\delta)$ we get $\Pr_{R,X_j}[|(X_j \cup R) \cap C| > t] \leq (\delta/n)^{24}$, and so (recalling the definition of $\cE_{\mathsf{sparse}}$) clearly by a union bound over every $C \in \cC \setminus \cC_{B/2}$ and every $j \in [p]$, we have
\begin{align} \label{eq:sparse}
    \Pr_{R,X_1,\ldots,X_p}[\cE_{\mathsf{sparse}}] \geq 1-\frac{\delta}{6\log(s^2)}
\end{align}
as desired. Thus, the correctness follows from \cref{eq:R}, \cref{eq:X}, and \cref{eq:sparse}. 

For the query complexity, each call to $\LS$ in line (19) is on a set of size at most $4 B \ln(n/\delta)$ (recall \cref{eq:RcupX}) where each set is of size at most $72\ln(n/\delta)$. Thus, by \Cref{lem:sparse}, the number of queries made in line (19) is $O(B \ln^3 (n/\delta))$ and each query is of size at most 
\[
\sqrt{\frac{4B \ln(n/\delta)}{72 \ln (n/\delta)}} = \sqrt{B/18} < s
\]
Since $\LS$ is called $p = O(\frac{n}{B} \ln(n/\delta))$ times during an iteration, the total number of such queries in an iteration is $O(n \ln^4 (n/\delta))$. \end{proof}
    
This completes the proof of \Cref{lem:medium}. \end{proof}

Finally, we show that the sets of size at most $2n/s^2$ are learned easily in the final stage of the algorithm.

\begin{claim} \label{clm:small} If $\tC = \cC_{2n/s^2}$ in line (23), then at the end of the algorithm's execution we have $\tC = \cC$ with probability $1-\delta$. Moreover, the number of queries made in line (24) is $O(\frac{n^2}{s^2}\ln(n/\delta))$ and every query is of size at most $s$. \end{claim}

\begin{proof} The proof follows immediately from \Cref{lem:sparse}. We are conditioning on $\tC = \cC_{2n/s^2}$ and so every unknown set is of size at most $2n/s^2$. Thus, by \Cref{lem:sparse} the number of queries made by $\LS$ in line (24) is $O(\frac{n^2}{s^2}\ln(n/\delta))$ and every query is of size at most $\sqrt{n/(2n/s^2)} \leq s$. \end{proof}

Finally, combining \Cref{clm:large}, \Cref{lem:medium}, and \Cref{clm:small}, at the end of the algorithm's execution we have $\tC = \cC$ with probability $1-3\delta$, and this completes the proof of \Cref{thm:NA-bounded}.

\subsection{Learning a Sparse, Partially Known Partition: Proof of \Cref{lem:sparse}} \label{sec:sparse}

Pseudocode for the algorithm is given in \Alg{sparse}. Let $K = \bigcup_{C \in \cK} C$ denote the set of points belonging to a known set.

\begin{algorithm}[H]
\caption{$\LS(\cK,c,\delta)$} \label{alg:sparse}
\textbf{Input:} (1) Subset query access to hidden partition $\cC$ over $U$ with $n$ points. (2) Known sets $\cK \subseteq \cC$ with promise that $|C| \leq c$ for every $C \in \cC \setminus \cK$. (3) Error probability $\delta > 0$\;

\textbf{Output:} A partition $\widetilde{\cC}$ of $U$, which is equal to $\cC$ with probability $1-\delta$\;
%\Comment{Remark on non-adaptivity: the choice of queries is completely oblivious to and independent of $\cK$.}
%\algorithmiccomment{\color{blue} Remark on non-adaptivity: the choice of queries is completely oblivious to and independent of $\cK$\color{black}} \\ 

Initialize $E \gets \emptyset$ and let $K = \bigcup_{C \in \cK} C$ denote the union of known sets\;

\textbf{Repeat} $2 c n \ln (n^2/\delta)$ times: \\ 

$\longrightarrow$ Let $I$ be a uniform random subset of $U$ of size $\floor{\sqrt{n/c}}$\;

$\longrightarrow$ Let $c_{I,K} = \countq(I \cap K)$ which we can compute without making any queries since the partition $\cK$ of $K$ is known\;

$\longrightarrow$ \textbf{Query} $I$ and if $\countq(I) - c_{I,K} = |I \setminus K|$ (i.e., $I \setminus K$ is an independent set), then $E \gets E \cup {I \setminus K \choose 2}$\;

$\backslash \backslash$ $\triangleright$ Note that here we use the full power of the $\countq$ query. In particular, we are not simply simulating independent set queries. This appears to be crucial to achieve non-adaptivity. $\backslash \backslash$

$\backslash \backslash$ $\triangleright$ Note that all queries can be specified without any knowledge of $\cK$. In particular, this allows $\NASQ$ to be implemented non-adaptively. $\backslash \backslash$

\textbf{Return} the graph $G(U \setminus K, \overline{E})$\;

\end{algorithm}

\begin{figure}
    \hspace*{5cm}
    \includegraphics[scale=0.3]{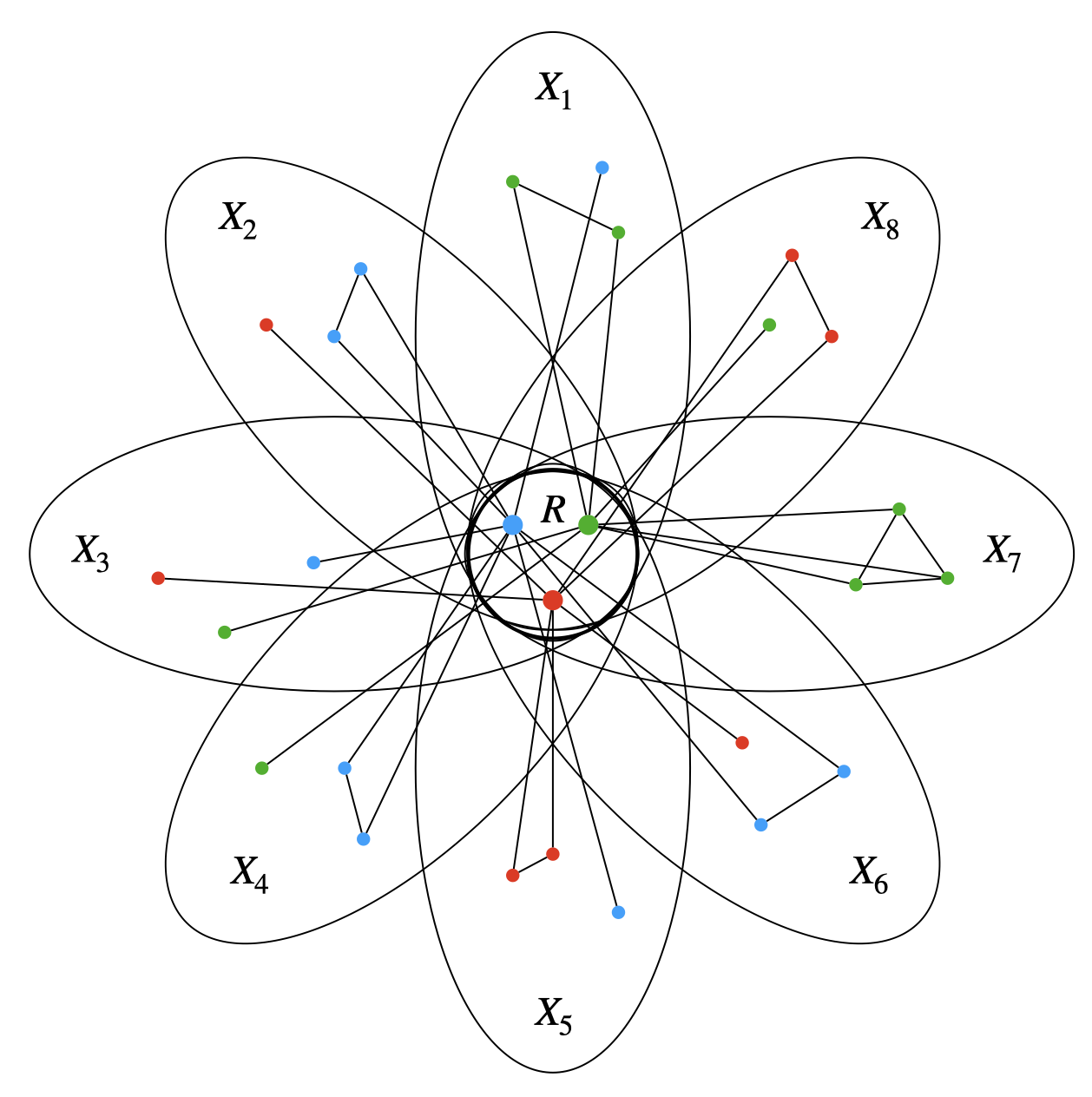}
    \caption{\small{An illustration depicting how \Alg{opt-bounded} learns the "medium"-sized sets. The sets of the partition are represented by the three colors (here $k=3$). The random core set $R$ (chosen in line 14) is large enough so that with high probability it contains a representative from every $C \in \cC_B$ which hasn't yet been learned. Then, we use \Alg{sparse} to learn the partition restricted on $X_j \cup R$ for every $X_j$ (chosen in line 15). The key is that these sets are small enough so that every unlearned set in the partition restricted on $X_j \cup R$ will be very small, allowing us to learn this restricted partition with few queries using \Alg{sparse}. Then, since the $X_j$-s cover $U$, we recover all unlearned sets in $\cC_B$ in lines 20-21 by taking all connected components of size at least $n/B$.}} 
    \label{fig:sunflower}
\end{figure}

\begin{figure}
    \hspace*{6cm}
    \includegraphics[scale=0.3]{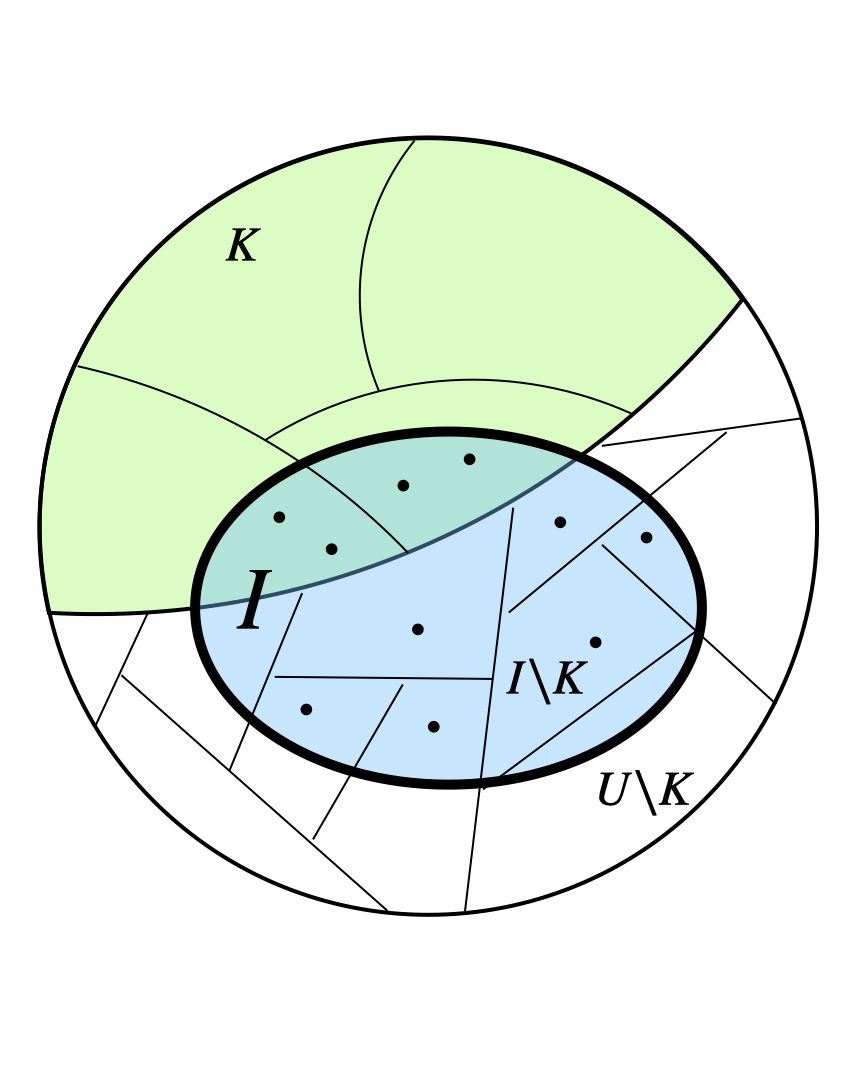}
    \caption{\small{An illustration depicting the key idea of \Alg{sparse} ($\LS$). The green region represents the points belonging to known sets in the partition (the larger sets) and the white regions represent the sets we are trying to learn. The queried set $I$ (pictured in blue) is drawn in line 5. Since the partition on the green region $K$ is known, the reconstruction process can simulate the oracle on the set $I \setminus K$ using the oracle's response to the set $I$ (line 6-7). In particular, the algorithm deduces that $I \setminus K$ is an independent set (line 7), even though $I$ itself is not, and can use this information to recover the unknown sets of the partition. } }
    \label{fig:IS-sim}
\end{figure}

The claimed query complexity and query size bound are obvious by definition of the algorithm. We now show that the connected components of the graph returned in line (8) correspond exactly with the collection of unknown sets $\cC \setminus \cK$ with probability $1-\delta$. 

Let us call a pair of points $x,y \in U \setminus K$ \emph{separated} if they belong to different, unknown sets. If a separated pair is contained in some $I$ for which $I \setminus K$ is an independent set (refer to \Alg{sparse} line (5)), then we learn that $x,y$ do not belong to the same set (i.e. we learn that this is a non-edge of the partition-graph over $U \setminus K$). When this occurs, we record that $x,y$ is a non-edge by adding it to $E$ in line (7). Therefore, if we learn this for \emph{every separated pair}, then we can recover the partition-graph. In particular, at the end of the algorithm's execution $E$ will be exactly the set of non-edges, and so $G(U \setminus K, \overline{E})$ will be partition-graph on $U \setminus K$. Thus, it suffices to show that every separated pair is contained in some $I$ for which $I \setminus K$ is an independent set. Importantly, to tell if $I \setminus K$ is an independent set it suffices to only query $I$ by exploiting the already known partition on $K$ (line 7).

Let $\ell =  2 c n \ln (n^2/\delta)$ be the total number of $I$-s that are sampled in line (5), and let them be denoted by $I_1,\ldots,I_{\ell}$. For each separated pair $x,y$ and for each $i \in [\ell]$, let $\cE_{i,x,y}$ denote the event that $x,y \in I_i$ and $I_i$ is an independent set ($\mathsf{IS}$). We will show the following claim, which completes the proof of \Cref{lem:sparse} by a union bound over all (at most $n^2$) separated pairs.

\begin{claim} For every separated pair $x,y$, we have
\[
\Pr_{I_1,\ldots,I_{\ell}}[\neg \cE_{1,x,y} \wedge \cdots \wedge \neg \cE_{\ell,x,y}] \leq \delta/n^2
\]
\end{claim}

\begin{proof} We will drop $x,y$ from the subscript for brevity. Let $I$ be a uniform random subset of $U$ of size $\floor{\sqrt{n/c}}$. We have $\Pr[\cE_i] = \Pr[x,y \in I] \cdot \Pr[I \setminus K \text{ is an } \mathsf{IS} ~|~ x,y \in I]$. Then, 
\[
\Pr[x,y \in I] = \frac{ {n-2 \choose |I|-2} } { {n \choose |I| } } = \frac{|I|(|I|- 1)}{n (n-1)} \geq \frac{1}{c n}\text{.}
\]
Now, recall that we are assuming every unknown set $C \in \cC \setminus \cK$ has size at most $c$ and so the probability of two points (either both uniform random, or one being $x$ or $y$ and the other being uniform random) landing in the same unknown set is at most $c/n$. Thus, by a union bound over all pairs of points $I$, we have
\[
\Pr[I \setminus K \text{ not an } \mathsf{IS} ~|~ x,y \in I] \leq {|I| \choose 2} \frac{c}{n} \leq 1/2
\]
and so $\Pr[\cE_i] \geq \frac{1}{2cn}$. Finally, since the $I_i$-s are independent, the $\cE_i$ events are also independent. Thus, 
\[
\Pr_{I_1,\ldots,I_{\ell}}[\neg \cE_{1,x,y} \wedge \cdots \wedge \neg \cE_{\ell,x,y}]  \leq \left(1 - \frac{1}{2 c n}\right)^{\ell} \leq \delta/n^2
\]
and this completes the proof. \end{proof}

\subsection{Low-Round Weak Subset Query Algorithm} \label{sec:LR-count}

In this section, we give a nearly optimal $r$-round algorithm for partition learning using weak subset queries of bounded size $s$. We prove the following theorem.

\begin{theorem} \label{thm:r-round-weak} Given $r \geq 1,\delta >0$, there is a randomized $r$-round algorithm for learning a $k$-partition using
\[
O\left(\frac{n^{\left(1 + \frac{1}{2^{r}-1}\right)} k^{\left(1 - \frac{1}{2^{r}-1}\right)}}{s^2} \cdot \log^5 (n/\delta)\right)
\]
weak subset queries of size at most $2 \leq s \leq \sqrt{n^{\left(\frac{1}{2^{r}-1}\right)} k^{\left(1 - \frac{1}{2^{r}-1}\right)}}$ that succeeds with probability $1-\delta$. 
\end{theorem}

In particular, with $O(\log \log n)$ rounds and query-size bound $\sqrt{k}$ our algorithm has query complexity $\widetilde{O}(n)$. In general, with $O(\log \log n)$ rounds and $s \leq \sqrt{k}$, we get $\widetilde{O}(\frac{nk}{s^2})$ queries, matching the $\Omega(\frac{nk}{s^2})$ fully adaptive lower bound up to $\poly \log n$ factors. Our main building block is the nearly optimal non-adaptive algorithm of \Cref{thm:NA-bounded}. Our $r$-round algorithm is obtained by combining the non-adaptive algorithm with the recursive approach used to obtain an optimal $r$-round pair-query algorithm in \Cref{sec:LR-pair}. For simplicity of presentation, we will use the following slightly weaker statement of \Cref{thm:NA-bounded}. (Note that in some cases the $\log$-factors can be improved in the non-adaptive algorithm and this also leads to the same improvements in the resulting $r$-round algorithm.)

\begin{theorem} [Corollary of \Cref{thm:NA-bounded}] \label{thm:NA-bounded-simp} There is a non-adaptive algorithm, $\NASQ$, which for any query size bound $2 \leq s \leq \sqrt{n}$ and error probability $\delta > 0$, learns an arbitrary partition on $n$ elements exactly with probability $1-\delta$ using at most $\frac{n^2}{s^2} \cdot C \log^5(n/\delta)$ weak subset queries of size at most $s$, where $C > 0$ is a universal constant. \end{theorem}

\begin{proofof}{\Cref{thm:r-round-weak}} For shorthand in the proof, we will use $\eps(r) = \frac{1}{2^{r}-1}$ for all $r \geq 1$. We will use the non-adaptive $s$-bounded query algorithm $\NASQ$ of \Cref{thm:NA-bounded-simp}, which uses at most
\begin{align} \label{eq:NA-queries}
    \frac{n^2}{s^2}  \cdot C \log^5(n/\delta)
\end{align}
queries of size at most $s$ where $s \leq [2,\sqrt{n}]$, and succeeds with probability $1-\delta$. The algorithm is recursive and pseudocode is given in \Alg{low-round-subset}.

\paragraph{Query complexity.} We prove by induction on $r$ that \Alg{low-round-subset} uses at most
\begin{align} \label{eq:r-weak-subset-query-complexity}
    \frac{n^{\left(1 + \frac{1}{2^{r}-1}\right)} k^{\left(1 - \frac{1}{2^{r}-1}\right)}}{s^2} \cdot 16C \log^5 (n/\delta)
\end{align}
weak subset queries of size at most $s$.  For the base case of $r=1$, $\LRSQ$ simply runs $\NASQ$, and so the base case is correct by \Cref{thm:NA-bounded-simp}.

\begin{algorithm}
\caption{$\LRSQ(U,s,\delta,r)$} \label{alg:low-round-subset}
\textbf{Input:} Subset query access (on subsets of size at most $s$) to hidden partition $\cC$ over $U$ with $n$ points and $|\cC| \leq k$. An allowed error probability $\delta$, and an allowed number of rounds $r$. Let $\eps(r) := \frac{1}{2^r-1}$\;
%\textbf{Output:} A partition $\widetilde{\cC}$ of $U$, which is equal to $\cC$ with probability $1-\delta$\;
\If{$r = 1$ \textbf{or} $n \leq 16k$ }{ 
    Run $\NASQ(U,s,\delta)$ and output the returned partition\; 
} \Else{

%\emph{(Round 1):} \\

Partition the $n$ points of $U$ arbitrarily into $\ell = \ceil{\frac{1}{3}\left(\frac{n}{k}\right)^{1-\eps(r)}}$ sets $U_1,\ldots,U_{\ell}$ each of size $|U_i| \leq t := \ceil{3n^{\eps(r)} k^{1-\eps(r)}}$\;

\For{$i\in [\ell]$} {
Run $\NASQ(U_i,s,\delta)$ and let $\widetilde{\cC}_i$ be the returned partition\;

$\backslash \backslash$ $\triangleright$ Note that $s \leq \sqrt{|U_i|}$ and so this call to $\NASQ$ is valid. $\backslash \backslash$

$\backslash \backslash$ $\triangleright$ Note that $\widetilde{\cC}_i = \{C \cap U_i \colon C \in \cC \}$ with probability $1-\delta$. $\backslash \backslash$
%to learn the partition $\cC_i = \{C \cap U_i \colon C \in \cC \}$\;

Form $R_i$ by taking exactly one representative from each set $C \in \widetilde{\cC}_i$\;
}
Set $R = R_1 \cup \cdots \cup R_{\ell}$ and $s' = \min(s,\sqrt{|R|^{\eps(r-1)} k^{1-\eps(r-1)}})$ (this is so that the recursive call to $\LRSQ$ has a valid query-size bound)\;
Recursively call $\LRSQ(R,s',\delta,r-1)$ and let $\widetilde{\cC}_R$ be the returned partition of $R$\;
\textbf{Return} the partition $\{\bigcup_{C' \in \cC_1 \cup \cdots \cup \cC_{\ell} \colon C' \cap C \neq \emptyset} C' \colon C \in \widetilde{\cC}_R\}$\;
}
\end{algorithm}

Now, suppose $r \geq 2$. First, if $n \leq 16k$, then the algorithm simply runs $\NASQ$ on $U$ using query size $s$ (line 3), for a total of
\[
\frac{n^2}{s^2} \cdot C\log^5(n/\delta) \leq \frac{nk}{s^2} \cdot 16C \log^5(n/\delta) \leq \frac{n^{1+\eps(r)}k^{1-\eps(r)}}{s^2} \cdot 16C \log^5(n/\delta)
\]
queries, where the last step simply used $n \geq k$. Note that in this case the algorithm is non-adaptive. This completes the proof for the case of $n \leq 16k$.

Now suppose $n > 16k$. Since $r \geq 2$, note that $1-\eps(r) \geq 1-\eps(2) \geq 2/3$. Using these bounds, we have $(n/k)^{1-\eps(r)} > 16^{2/3} > 6$. Recalling the definition of $\ell$ and $t$ in line (6), this implies that
\begin{align} %\label{eq:round_1-weak}
    \ell \leq \frac{1}{3}\left(\frac{n}{k}\right)^{1-\eps(r)} + 1 = \frac{1}{2}\left(\frac{n}{k}\right)^{1-\eps(r)} - \frac{1}{6}\left(\frac{n}{k}\right)^{1-\eps(r)} + 1 < \frac{1}{2}\left(\frac{n}{k}\right)^{1-\eps(r)}
\end{align}
and clearly $t \leq 4n^{\eps(r)}k^{1-\eps(r)}$. Then, using the bound \cref{eq:NA-queries} on the query complexity of $\NASQ$, the first round (lines 7-9) makes at most
\begin{align} \label{eq:round_1-weak}
    \ell \cdot \frac{t^2}{s^2} C\log^5(n/\delta) &< \frac{1}{2}\left(\frac{n}{k}\right)^{1-\eps(r)} \cdot \left(\frac{4n^{\eps(r)} k^{1-\eps(r)}}{s} \right)^2 \cdot C \log^5(n/\delta) \nonumber \\
    &= \frac{n^{1+\eps(r)} k^{1-\eps(r)}}{s^2} \cdot 8C \log^5 (n/\delta)
\end{align}
queries, all of size at most $s$. Then, the resulting set $R$ is of size
\begin{align} \label{eq:|R|bound}
    |R| \leq k \cdot \frac{1}{2} \left(\frac{n}{k}\right)^{1-\eps(r)} = \frac{1}{2} \cdot n^{1-\eps(r)} k^{\eps(r)}\text{.}
\end{align}

%\comment{Finish from here.}

We now show that the total number of queries made over the remaining $r-1$ rounds is at most $\frac{n^{1+\eps(r)} k^{1-\eps(r)}}{s^2} \cdot 8C \log^5 (n/\delta)$. Combining this with \cref{eq:round_1-weak} shows that the total number of queries is at most $\frac{n^{1+\eps(r)} k^{1-\eps(r)}}{s^2} \cdot 16C \log^5 (n/\delta)$, which completes the proof. 

First, suppose that $s' = \sqrt{|R|^{\eps(r-1)} k^{1-\eps(r-1)}}$. By induction, the recursive call in line (12) costs at most
\begin{align*} %\label{eq:round_r-1'}
    \frac{|R|^{1+\eps(r-1)}k^{1-\eps(r-1)}}{|R|^{\eps(r-1)}k^{1-\eps(r-1)}} \cdot 16 C \log^5 (n/\delta) &= |R| \cdot 16C \log^5 (n/\delta) \\
    &\leq 8C n \log^5 (n/\delta) \leq \frac{n^{1+\eps(r)} k^{1-\eps(r)}}{s^2} \cdot 8C \log^5 (n/\delta)
\end{align*}
queries, where we used the bound on $|R|$ and the fact that $k \leq n$. %Adding up the number of queries made in \cref{eq:round_1} and \cref{eq:round_r-1'} completes the proof.
Now, suppose $s' = s$. By induction, using the upper on $|R|$ from \cref{eq:|R|bound}, we obtain that running $\LRSQ(R,s,\delta,r-1)$ costs at most
\begin{align} \label{eq:round_r-1-weak}
    &\frac{\left(\frac{1}{2} \cdot n^{1-\eps(r)} k^{\eps(r)}\right)^{1+\eps(r-1)} k^{1-\eps(r-1)}}{s^2} \cdot 16C \log^5(n/\delta)  \nonumber \\ 
    &\leq \frac{n^{(1-\eps(r))(1+\eps(r-1))} k^{\eps(r)(1+\eps(r-1))+(1-\eps(r-1))}}{s^2} \cdot 8C \log^5 (n/\delta) \nonumber \\
    &= \frac{n^{1+\eps(r)}k^{1-\eps(r)}}{s^2} \cdot 8C \log^5 (n/\delta)
\end{align}
queries where the equality is by \Cref{clm:eps(r)}. This completes the proof of the query complexity. 

\paragraph{Correctness.} Observe that the total number of calls to $\NASQ$ per round in \Alg{low-round-subset} is clearly upper bounded by $n$. Thus, by a union bound all such calls are successful with probability at least $1-rn\delta$. We can then simply set $\delta = \delta/rn$, only increasing the query complexity by a constant factor. Now, conditioned on all calls to $\NASQ$ being successful, the proof of correctness for $\LRSQ$ is simple and completely analogous to that of \Cref{thm:LR-pair-UB}. This completes the proof of \Cref{thm:r-round-weak}. \end{proofof}

\section{Strong Subset Queries} \label{sec:LR-partition}

In this section, we investigate the query complexity of learning a $k$-partition over $|U| = n$ elements via access to a strong subset query oracle. The algorithm is given a query size bound $s$ and can query the oracle on any set $S \subseteq U$ of size $|S| \leq s$. For hidden partition $\cP$, the oracle returns $(X \cap S \colon X \in \cP)$, the partition restricted on $S$. We first design a simple deterministic non-adaptive algorithm in \Cref{thm:NA-strong-UB}, which is optimal among all (even randomized) non-adaptive algorithms. We then use this result to obtain an $r$-round deterministic algorithm in \Cref{thm:LR-strong} which is (nearly) optimal for all $r$.

\begin{theorem} \label{thm:NA-strong-UB} For any even $2 \leq s \leq n$, there is a deterministic non-adaptive algorithm for learning an arbitrary partition using at most $\frac{5n^2}{s^2}$ strong subset queries of size $s$. \end{theorem}

\begin{proof} It suffices to design a collection of sets of size at most $s$ such that for every pair of points $x,y$ there is some subset containing both $x$ and $y$. We claim that the following construction has this property. Partition the $n$ points into $\ell = \ceil{\frac{n}{s/2}} \leq \frac{3n}{s}$ sets (since $s \leq n$) $T_1,\ldots,T_{\ell}$ of size $|T_i| = s/2$. Then the set of queries is $T_i \cup T_j$ for every $i \neq j$ for a total of $\smash{{\ell \choose 2} < \frac{5n^2}{s^2}}$ queries of size exactly $s$. If $x,y$ are in the same $T_i$, then clearly this pair is covered by any query involving $T_i$, and if $x \in T_i, y \in T_j$ for $i \neq j$, then clearly this pair is covered by the query $T_i \cup T_j$. \end{proof}

\begin{theorem} \label{thm:LR-strong} For any $r \geq 1$, there is a deterministic $r$-round algorithm for learning a $k$-partition using
\[
\frac{80 n^{\left(1 + \frac{1}{2^{r}-1}\right)} k^{\left(1 - \frac{1}{2^{r}-1}\right)}}{s^2} \text{ strong subset queries of size at most } s \leq n^{\left(\frac{1}{2^r-1}\right)}k^{\left(1-\frac{1}{2^r-1}\right)} \text{ where } s \geq 2 
\text{ is even.}
\]

\end{theorem}

%\comment{need to add floors and ceilings in the algorithm and analysis}

In particular, with $O(\log \log n)$ rounds and $s \leq k$, we get $O(\frac{nk}{s^2})$ queries, matching the $\Omega(\max(\frac{nk}{s^2},\frac{n}{s}))$ fully adaptive lower bound. With $O(\log \log n)$ rounds we get an algorithm making $O(\frac{n}{k})$ strong subset queries of size $k$. To design our $r$-round algorithm we essentially follow the same strategy presented in \Cref{sec:LR-pair} for pair queries, replacing the non-adaptive subroutine with that of \Cref{thm:NA-strong-UB}.

\begin{proof}[Proof of \Cref{thm:LR-strong}] For shorthand in the proof, we will use $\eps(r) = \frac{1}{2^{r}-1}$ for all $r \geq 1$. We will use the non-adaptive $s$-bounded strong subset query algorithm $\NASQS$ of \Cref{thm:NA-strong-UB}. By \Cref{thm:NA-strong-UB}, this algorithm learns a $k$-partition of $n$ elements using at most $\frac{5n^2}{s^2}$ queries of size at most $s$ where $s \leq [2,n]$ is even.

The algorithm is recursive and pseudocode is given in \Alg{low-round-subset-strong}. We prove the theorem by induction on $r$. For the base case of $r=1$, $\LRSQS$ simply runs $\NASQS$, and correctness follows by \Cref{thm:NA-strong-UB}.

\begin{algorithm}
\caption{$\LRSQS(U,s,r)$} \label{alg:low-round-subset-strong}
\textbf{Input:} Strong subset query access (on subsets of size at most $s$) to hidden partition $\cC$ over $U$ with $n$ points and $|\cC| \leq k$. An allowed number of rounds $r$. Let $\eps(r) = \frac{1}{2^r-1}$\;

\If{$r = 1$ \textbf{or} $n \leq 16k$}{ 
    Run $\NASQS(U,s)$ and output the returned partition\; 
} \Else{

Partition the $n$ points of $U$ arbitrarily into $\ell = \ceil{\frac{1}{3}\left(\frac{n}{k}\right)^{1-\eps(r)}}$ sets $U_1,\ldots,U_{\ell}$ each of size $|U_i| \leq t := \ceil{3n^{\eps(r)} k^{1-\eps(r)}}$\;

\For{$i\in [\ell]$} {
Run $\NASQS(U_i,s)$ to learn the partition $\cC_i = \{C \cap U_i \colon C \in \cC \}$\;
$\backslash \backslash$ $\triangleright$ Note that $s \leq |U_i|$ and so this is a valid call to $\NASQS$. $\backslash \backslash$ \\
%\Comment{\color{blue} Note that $\widetilde{\cC}_i = \{C \cap U_i \colon C \in \cC \}$ with probability $1-\delta$ \color{black}} 
%to learn the partition $\cC_i = \{C \cap U_i \colon C \in \cC \}$\;
Form $R_i$ by taking exactly one representative from each set $C \in \cC_i$\;
}
Set $R = R_1 \cup \cdots \cup R_{\ell}$ and $s' = \min(s,|R|^{\eps(r-1)} k^{1-\eps(r-1)})$ (this is so that the recursive call to $\LRSQS$ has a valid query-size bound)\;
Recursively call $\LRSQ(R,s',r-1)$ and let $\cC_R$ be the returned partition of $R$\;
\textbf{Return} the partition $\{\bigcup_{C' \in \cC_1 \cup \cdots \cup \cC_{\ell} \colon C' \cap C \neq \emptyset} C' \colon C \in \cC_R\}$\;
}
\end{algorithm}

Now suppose $r \geq 2$. First, if $n \leq 16k$, then we simply run $\NASQS$ on $U$ using query size $s$ (line 3), for a total of
\[
\frac{5n^2}{s^2} \leq \frac{80 nk}{s^2} \leq \frac{80 n^{1+\eps(r)}k^{1-\eps(r)}}{s^2}
\]
queries, where the last step simply used $n \geq k$. Note that in this case the algorithm is non-adaptive and again correctness follows from \Cref{thm:NA-strong-UB}. This completes the proof for the case of $n \leq 16k$.

Now suppose that $n > 16k$. The proof of correctness is simple and completely analogous to that of \Cref{thm:LR-pair-UB}. Now let us prove the desired query complexity. Since $r \geq 2$, note that $1-\eps(r) \geq 1-\eps(2) \geq 2/3$. Using these bounds, we have $(n/k)^{1-\eps(r)} > 16^{2/3} > 6$. Recalling the definition of $\ell$ and $t$ in line (6), this implies that
\begin{align} \label{eq:round_1-strong}
    \ell \leq \frac{1}{3}\left(\frac{n}{k}\right)^{1-\eps(r)} + 1 = \frac{1}{2}\left(\frac{n}{k}\right)^{1-\eps(r)} - \frac{1}{6}\left(\frac{n}{k}\right)^{1-\eps(r)} + 1 < \frac{1}{2}\left(\frac{n}{k}\right)^{1-\eps(r)}
\end{align}
and clearly $t \leq 4n^{\eps(r)}k^{1-\eps(r)}$. Thus, using the bound on the query complexity of $\NASQS$, the first round (line 8) makes at most
\begin{align} \label{eq:round_1-strong'}
     \ell \cdot \frac{5t^2}{s^2} < \frac{1}{2}\left(\frac{n}{k}\right)^{1-\eps(r)} \cdot \frac{80 n^{2\eps(r)} k^{2(1-\eps(r))}}{s^2} \leq \frac{40n^{1+\eps(r)} k^{1-\eps(r)}}{s^2}
\end{align}
queries. Then, the resulting set $R$ is of size
\begin{align} \label{eq:|R|bound-strong}
    |R| \leq k \cdot \frac{1}{2} \left(\frac{n}{k}\right)^{1-\eps(r)} = \frac{1}{2} \cdot n^{1-\eps(r)} k^{\eps(r)}\text{.}
\end{align}
We now show that the total number of queries in the remaining $r-1$ rounds is at most $\frac{40 n^{1+\eps(r)} k^{1-\eps(r)}}{s^2}$. Combining this with \cref{eq:round_1-strong} shows that the total number of queries is then at most this $\frac{80 n^{1+\eps(r)} k^{1-\eps(r)}}{s^2}$, which completes the proof.

First, suppose that $s' = |R|^{\eps(r-1)} k^{1-\eps(r-1)}$. By induction, the recursive call to $\LRSQS$ in line (12) costs at most
\begin{align} \label{eq:round_r-1'-strong}
    80 \frac{|R|^{1+\eps(r-1)}k^{1-\eps(r-1)}}{|R|^{2\eps(r-1)}k^{2(1-\eps(r-1))}} &= 80 \left(\frac{|R|}{k}\right)^{1-\eps(r-1)} 
\end{align}
queries. If $|R| \leq k$, then the RHS above is at most $80$, which clearly satisfies the desired bound. Otherwise, the RHS of \cref{eq:round_r-1'-strong} is at most $80|R|/k$ and using \cref{eq:|R|bound-strong}, this is at most $40(n/k)^{1-\eps(r)}$. Now, since $s \leq n^{\eps(r)}k^{1-\eps(r)}$, observe that 
\[
\frac{40 n^{1+\eps(r)} k^{1-\eps(r)}}{s^2} \geq 40 \left(\frac{n}{k}\right)^{1-\eps(r)} 
\]
and so the number of queries in the remaining $r-1$ rounds satisfies the desired bound.
%Adding up the number of queries made in \cref{eq:round_1-strong} and \cref{eq:round_r-1'-strong} completes the proof.

Now, suppose $s' = s$. By induction, using the upper on $|R|$ from \cref{eq:|R|bound-strong}, we obtain that running $\LRSQS(R,s,r-1)$ costs at most
\begin{align*} \label{eq:round_r-1-strong}
    &\frac{80}{s^2} \left(\frac{1}{2} \cdot n^{1-\eps(r)} k^{\eps(r)}\right)^{1+\eps(r-1)} k^{1-\eps(r-1)} \leq \frac{40 n^{(1-\eps(r))(1+\eps(r-1))} k^{\eps(r)(1+\eps(r-1))+(1-\eps(r-1))}}{s^2}  = \frac{40n^{1+\eps(r)}k^{1-\eps(r)} }{s^2} 
\end{align*}
queries where the equality is by \Cref{clm:eps(r)}. %Combining \cref{eq:round_1-strong} and \cref{eq:round_r-1-strong} establishes the claimed query complexity. 
This completes the proof. \end{proof}

% Acknowledgments---Will not appear in anonymized version
%\acks{We thank a bunch of people and funding agency.}

\bibliographystyle{alpha}
\bibliography{biblio}

\end{document}